\documentclass[11pt,reqno]{article}
\usepackage{amssymb,amscd,amsmath,amsthm,color}
\newtheorem{defin}{Definition}[section]
\newtheorem{thm}[defin]{Theorem}
\newtheorem{example}[defin]{Example}
\newtheorem{cor}[defin]{Corollary}

\newtheorem{lemma}[defin]{Lemma}
\newtheorem{remark}[defin]{Remark}
\newtheorem{prop}[defin]{Proposition}

\def\bC{\mathbb{C}}
\def\bR{\mathbb{R}}
\def\cB{\mathcal{B}}
\def\B{\mathcal{B}}
\def\cH{\mathcal{H}}
\def\hil{\mathcal{H}}
\def\bN{\mathbb{N}}
\def\N{\mathbb{N}}
\def\R{\mathbb{R}}
\def\cA{\mathcal{A}}
\def\A{\mathcal{A}}

\def\ffi{\varphi}
\def\cE{\mathcal{E}}
\def\eps{\varepsilon}
\def\bT{\mathbb{T}}
\def\bZ{\mathbb{Z}}
\def\bz{\left(}
\def\jz{\right)}
\def\supp{\mathrm{supp}\,}
\def\<{\langle}
\def\>{\rangle}
\def\ds{\mbox{ }\mbox{ }}

\def\ki{\textit}
\def\vfi{\varphi}
\def\kii{\textit}
\def\kiii{}
\def\ep{\varepsilon}
\def\psin{\psi^\circ}

\newcommand{\norm}[1]{\|#1\|}

\newcommand{\D}[1]{\hat{#1}}
\newcommand{\derleft}[1]{\partial^{-} #1}
\newcommand{\derright}[1]{\partial^{+} #1}
\newcommand{\hli}[3]{\underline h_{\scriptscriptstyle{G}}(#3\,|\,#1\,\|\,#2)}
\newcommand{\hls}[3]{\overline h_{\scriptscriptstyle{G}}(#3\,|\,#1\,\|\,#2)}
\newcommand{\hl}[3]{h_{\scriptscriptstyle{G}}(#3\,|\,#1\,\|\,#2)}
\newcommand{\cli}[2]{\underline c_{\scriptscriptstyle{G}}(#1 , #2)}
\newcommand{\cls}[2]{\overline c_{\scriptscriptstyle{G}}(#1, #2)}
\newcommand{\cl}[2]{c_{\scriptscriptstyle{G}}(#1,#2)}
\newcommand{\stli}[2]{\underline s_{\scriptscriptstyle{G}}(#1\,\|\,#2)}
\newcommand{\stls}[2]{\overline s_{\scriptscriptstyle{G}}(#1\,\|\,#2)}
\newcommand{\stl}[2]{s_{\scriptscriptstyle{G}}(#1\,\|\,#2)}
\newcommand{\chdist}[2]{C(#1,#2)}
\newcommand{\chmdist}[2]{C_M(#1,#2)}
\newcommand{\hdist}[3]{H(#3\,|\,#1\,\|\,#2)}
\newcommand{\hmdist}[3]{H_M(#3\,|\,#1\,\|\,#2)}
\newcommand{\sr}[2]{S\bz #1\,\|\, #2\jz}
\newcommand{\msr}[2]{S_M\bz #1\,\|\, #2\jz}

\newcommand{\pr}[1]{|#1\rangle\langle #1|}

\DeclareMathOperator{\Tr}{Tr}
\DeclareMathOperator{\Ad}{Ad}

\begin{document}
\allowdisplaybreaks

\ \vskip 1cm
\centerline{\LARGE {\bf Quantum hypothesis testing}}
\centerline{\LARGE {\bf with group symmetry}}
\bigskip
\bigskip
\centerline{\Large
Fumio Hiai,\footnote{E-mail: hiai@math.is.tohoku.ac.jp}
Mil\'an Mosonyi\footnote{E-mail: milan.mosonyi@gmail.com}
and Masahito Hayashi\footnote{E-mail: hayashi@math.is.tohoku.ac.jp}}
	
\medskip
\begin{center}
$^{1,\,3}$\,Graduate School of Information Sciences, Tohoku University \\
Aoba-ku, Sendai 980-8579, Japan
\end{center}
\begin{center}
$^2$\,Mathematical Institute, Budapest University of Technology and Economics \\
Egry J\'ozsef u~1., Budapest, 1111 Hungary
\end{center}

\medskip
\begin{abstract}
The asymptotic discrimination problem of two quantum states is studied in the setting
where measurements are required to be invariant under some symmetry group of the system.
We consider various asymptotic error exponents in connection with the problems of the
Chernoff bound, the Hoeffding bound and Stein's lemma, and derive bounds on these
quantities in terms of their corresponding statistical distance measures. A special
emphasis is put on the comparison of the performances of group-invariant and unrestricted
measurements.
 

\end{abstract}

\medskip
\section{Introduction}

In the asymptotic framework of (quantum) state discrimination, one is provided with
several copies of a quantum system and with the knowledge that the state of the system
is either $\rho_0$ (null hypothesis $H_0$) or $\rho_1$ (alternative hypothesis $H_1$).
One's aim is to decide, based on measurements on the copies, which one the true state is.
For simplicity, we will assume here that the Hilbert space $\hil$ of the system is finite
dimensional, and hence the states can be represented by density operators $\D{\rho}_k$
that satisfy $\rho_k(A)=\Tr\D{\rho}_kA$ for any observable $A\in\B(\hil)$ and $k=1,2$.
A measurement on $n$ copies is given by a binary positive operator valued measure (POVM)
$(T,I-T)$ with $T\in\B(\hil^{\otimes n}),\,0\le T\le I$, where $T$ corresponds to
accepting $\rho_0$ and $I-T$ to accepting $\rho_1$. An erroneous decision is made if
$H_0$ is accepted when it is false (error of the first kind) or the other way around
(error of the second kind). The probabilities of these events are given by 
$$
\beta_{0,n}(T):=\rho_0^{\otimes n}(I-T),\qquad
\beta_{1,n}(T):=\rho_1^{\otimes n}(T).
$$

The optimal asymptotic performance in a state discrimination problem can be defined in
various ways, depending on whether or not the two hypotheses are treated as of equal
importance. Usually, one is interested in the exponential decay rates of the above error
probabilities or combinations of them, in the $n\to\infty$ limit. The most studied
quantities are the following:
\begin{itemize}
\item[(i)] the optimal exponential decay rate of the sum of the two kinds of error
probabilities (Chernoff bound),
\item[(ii)] the optimal exponential decay rate of the error probabilities of the second
kind under the assumption that the error probabilities of the first kind decay with a
given exponential speed (Hoeffding bound),
\item[(iii)] the optimal exponential decay rate of the error probabilities of the second
kind under the assumption that the error probabilities of the first kind vanish
asymptotically (Stein's lemma).
\end{itemize}

The quantum problem of Stein's lemma was solved in \cite{HP,ON} (see also \cite{Ha0}),
where it was shown that the optimal error bound is equal to the relative entropy of the
two states, hence providing an operational interpretation of the relative entropy.
Recently, the solution of the quantum problem of the Chernoff bound \cite{Aud,NSz}
created a renewed interest in hypothesis testing problems. The techniques developed in
\cite{Aud,NSz} were also used in \cite{Hayashi,Nagaoka} to solve the quantum problem of
the Hoeffding bound, improving a weaker bound previously given in \cite{OH}. The optimal
error bounds in these cases are the Chernoff distance and the Hoeffding distance,
respectively. All these results deal with the case where one is allowed to perform any
collective measurement to discriminate i.i.d.\ (independent and identically distributed)
extensions of the states $\rho_0$ and $\rho_1$. Various extensions
to non-i.i.d.\ scenarios were also treated in the works
\cite{BSC,BDKSSCSz,HMO1,HMO,HP2,MHOF,M}. Note that the present formulation describes
only the simple hypothesis testing problem, i.e., when both the null and the alternative
hypotheses are a single state of the system. Some results in the case where one of the
hypotheses is composite (i.e., a subset of the state space) were obtained e.g., in
\cite{BDKSSCSz1,BDKSSCSz,BP}.

The purpose of the present paper is to treat the optimal error exponents (i), (ii), and
(iii) in the case where the states to discriminate are still i.i.d.\ extensions of the
two simple hypotheses but measurements are restricted to those invariant under the action
of some symmetry group of the system. As symmetries and dynamics are described in the
same way in the algebraic formalism, this setting also contains the case where one is
only able to measure functions of the energy. Indeed, the group in this case is the
dynamical group generated by the Hamiltonian of the system, and invariant measurements
are exactly those that commute with the Hamilton operator. Hypothesis testing with
group-invariant measurements has applications to the entanglement testing problem, as it
was shown in \cite{H-ent}.

The structure of the paper is as follows. In Section \ref{sec:formulations} we give a
detailed formulation of the problem. As it was shown in \cite{HMO}, the key to solve the
state discrimination problems is to determine the asymptotic R\'enyi relative entropies.
This is carried out for the present scenario in Section \ref{sec:distances}, and the
results are used in Section \ref{sec:error bounds} to give bounds on the various error
exponents. In particular, we provide a complete solution to the problem of Stein's lemma.
In Section \ref{sec:invariant} we analyze the case where the alternative hypothesis is
invariant under the symmetry group and in Section \ref{sec:examples} we show some
examples to compare the performances of restricted and unrestricted measurements. 

\section{Formulation of the problem}\label{sec:formulations}
\setcounter{equation}{0}

Let $\hil$ be a finite-dimensional Hilbert space with $d:=\dim\hil$  and let $\Tr$ be the
usual trace on $\B(\hil)$. Let $u:\,G\to\B(\hil)$ be a unitary representation $u$ of a
group $G$ on $\cH$. Since $G$ can be replaced without loss of generality by the closure of
$\{u_g:g\in G\}$ in the unitary group of $\B(\hil)$, we may and do assume that $G$ is a
compact group.  For each $n\in\bN$ consider the $n$-fold tensor product representation
$u^{\otimes n}:g\mapsto u_g^{\otimes n}\in \B(\hil)^{\otimes n},\,g\in G$, and define a
subalgebra $\cA_n$ of $\B(\hil)^{\otimes n}=\cB(\cH^{\otimes n})$ as the commutant of
$u_g^{\otimes n}$, $g\in G$, i.e.,
$$
\cA_n:=\{A\in \B(\hil)^{\otimes n}:Au_g^{\otimes n}=u_g^{\otimes n}A,\,g\in G\}.
$$
That is, $\cA_n$ is the fixed point subalgebra $(\B(\hil)^{\otimes n})^G$ of
$\B(\hil)^{\otimes n}$ under the action
$\mathrm{Ad}\,u_g^{\otimes n}:=u_g^{\otimes n}\cdot u_g^{*\otimes n}$,
$g\in G$. Then $\cA_1\subset\cA_2\subset\cdots$ by natural inclusions. Note that
$\cA_n\otimes\cA_m\subset\cA_{n+m}$ for any $n,m\in\N$. In particular,
$\cA_1^{\otimes n}\subset \cA_n$, $n\in\N$.

Let $E_{\cA_n}$ be the conditional expectation from $\B(\hil)^{\otimes n}$ onto $\cA_n$
with respect to the trace $\Tr$. Note that $E_{\cA_n}$ can be written in the integral form
\begin{equation}\label{F-2.1}
E_{\cA_n}(X)=\int_Gu_g^{\otimes n}Xu_g^{*\otimes n}\,dg,\qquad X\in\B(\hil)^{\otimes n},
\end{equation}
where $dg$ is the Haar probability measure on $G$. Let $\widehat G$ denote the
representation ring consisting of all unitary equivalence classes of irreducible
representations of $G$. For each $n\in\bN$ the $n$-fold tensor product representation
$u^{\otimes n}$ is decomposed into irreducible components as
$$
u^{\otimes n}=m_1^{(n)}u_1^{(n)}\oplus m_2^{(n)}u_2^{(n)}
\oplus\dots\oplus m_{k_n}^{(n)}u_{k_n}^{(n)},
$$
where $u_i^{(n)}\in\widehat G$, $1\le i\le k_n$, are contained in $u^{\otimes n}$ with
multiplicities $m_i^{(n)}$. For $1\le i\le k_n$ let $d_i^{(n)}$ be the dimension of
$u_i^{(n)}$ so that we have $\sum_{i=1}^{k_n}m_i^{(n)}d_i^{(n)}=d^n$ and we can identify
$\A_n$ with
\begin{equation}\label{F-2.2}
\cA_n=\bigoplus_{i=1}^{k_n}\Bigl(M_{m_i^{(n)}}\otimes I_{d_i^{(n)}}\Bigr),
\end{equation}
where $M_k:=\B(\bC^k),\,k\in\N$.
Then the conditional expectation $E_{\cA_n}:\,\B(\hil)^{\otimes n}\to\cA_n$ given in
\eqref{F-2.1} is rewritten as
\begin{equation}\label{F-2.3}
E_{\cA_n}(X)=\sum_{i=1}^{k_n}E_i^{(n)}(P_i^{(n)}XP_i^{(n)}),
\qquad X\in\B(\hil)^{\otimes n},
\end{equation}
where $P_i^{(n)}$ is the orthogonal projection onto the subspace corresponding to
$m_i^{(n)}u_i^{(n)}$ in the decomposition \eqref{F-2.2}, i.e., $P_i^{(n)}$ is the identity
$I_{m_i^{(n)}}\otimes I_{d_i^{(n)}}$ of $M_{m_i^{(n)}}\otimes I_{d_i^{(n)}}$, and
$E_i^{(n)}$ is the partial trace or the conditional expectation from
$M_{m_i^{(n)}}\otimes M_{d_i^{(n)}}$ onto $M_{m_i^{(n)}}\otimes I_{d_i^{(n)}}$ with respect
to the trace. As is well known (see \cite{Oh} for a detailed proof), the
representation ring of any compact group has polynomial growth so that we have
\begin{equation}\label{F-2.4}
\lim_{n\to\infty}{1\over n}\log\Biggl(\sum_{i=1}^{k_n}d_i^{(n)}\Biggr)=0.
\end{equation}

Consider now the hypothesis testing problem with null hypothesis $\rho_0$ and alternative
hypothesis $\rho_1$, as described in the Introduction. We will be interested in the
quantities
\begin{align}
\cli{\rho_0}{\rho_1}&:=\inf_{\{T_n\}}\biggl\{\liminf_{n\to\infty}
{1\over n}\log\bz\beta_{0,n}(T_n)+\beta_{1,n}(T_n)\jz\biggr\},\label{chernoff1}\\
\cls{\rho_0}{\rho_1}&:=\inf_{\{T_n\}}\biggl\{\limsup_{n\to\infty}
{1\over n}\log\bz\beta_{0,n}(T_n)+\beta_{1,n}(T_n)\jz\biggr\},\label{chernoff2}\\
\cl{\rho_0}{\rho_1}&:=\inf_{\{T_n\}}\biggl\{\lim_{n\to\infty}
{1\over n}\log\bz\beta_{0,n}(T_n)+\beta_{1,n}(T_n)\jz\biggr\},\label{chernoff3}
\end{align}
corresponding to the problem of the \ki{Chernoff bound}, 
\begin{align}
\hli{\rho_0}{\rho_1}{r}&:=\inf_{\{T_n\}}\biggl\{\liminf_{n\to\infty}
{1\over n}\log\beta_{1,n}(T_n)\,\bigg|\,
\limsup_{n\to\infty}{1\over n}\log\beta_{0,n}(T_n)<-r\biggr\},\quad r\ge0,
\label{hoeffding1}\\
\hls{\rho_0}{\rho_1}{r}
&:=\inf_{\{T_n\}}\biggl\{\limsup_{n\to\infty}
{1\over n}\log\beta_{1,n}(T_n)\,\bigg|\,
\limsup_{n\to\infty}{1\over n}\log\beta_{0,n}(T_n)<-r\biggr\},\quad r\ge0,
\label{hoeffding2}\\
\hl{\rho_0}{\rho_1}{r}
&:=\inf_{\{T_n\}}\biggl\{\lim_{n\to\infty}
{1\over n}\log\beta_{1,n}(T_n)\,\bigg|\,
\limsup_{n\to\infty}{1\over n}\log\beta_{0,n}(T_n)<-r\biggr\},\quad r\ge0,
\label{hoeffding3}
\end{align}
corresponding to the problem of the \ki{Hoeffding bound}, and
\begin{align}
\stli{\rho_0}{\rho_1}&:=\inf_{\{T_n\}}\biggl\{\liminf_{n\to\infty}
{1\over n}\log\beta_{1,n}(T_n)\,\bigg|\,
\lim_{n\to\infty}{1\over n}\log\beta_{0,n}(T_n)=0\biggr\}, \label{stein1}\\
\stls{\rho_0}{\rho_1}
&:=\inf_{\{T_n\}}\biggl\{\limsup_{n\to\infty}
{1\over n}\log\beta_{1,n}(T_n)\,\bigg|\,
\lim_{n\to\infty}{1\over n}\log\beta_{0,n}(T_n)=0\biggr\}, \label{stein2}\\
\stl{\rho_0}{\rho_1}
&:=\inf_{\{T_n\}}\biggl\{\lim_{n\to\infty}
{1\over n}\log\beta_{1,n}(T_n)\,\bigg|\,
\lim_{n\to\infty}{1\over n}\log\beta_{0,n}(T_n)=0\biggr\},\label{stein3}
\end{align}
corresponding to the problem of \ki{Stein's lemma}. Here, the infima are taken over
sequences of $G$-invariant measurements
$\{T_n\}_{n\in\N}$ with $T_n\in\A_n,\,0\le T_n\le I$. Note that posing $G$-invariance
on the measurements to distinguish $\rho_0^{\otimes n}$ from $\rho_1^{\otimes n}$ is
equivalent to considering the discrimination of the $G$-invariant states
\begin{equation}\label{F-2.5}
\rho_{0,n}:=\rho_0^{\otimes n}\circ E_{\A_n},\ds\ds\ds\rho_{1,n}
:=\rho_1^{\otimes n}\circ E_{\A_n}
\end{equation}
with unrestricted measurements, as we have
\begin{equation*}
\beta_{0,n}(E_{\A_n}(T_n))=\rho_0^{\otimes n}(I-E_{\A_n}(T_n))
=\rho_0^{\otimes n}(E_{\A_n}(I-T_n))=\rho_{0,n}(I-T_n)
\end{equation*}
and similarly for $\beta_{1,n}$. Hence, the asymptotic problem with $G$-invariant
measurements is equivalent to the asymptotic state discrimination problem of the two
sequences of $G$-invariant states $\{\rho_{0,n}\}_{n\in\bN}$ and
$\{\rho_{1,n}\}_{n\in\bN}$. Note also that the families $\{\rho_{k,n}\}_{n\in\N},\,k=0,1$
are compatible in the sense that $\rho_{k,m}|_{\B(\hil)^{\otimes n}}=\rho_{k,n},\,m\ge n$.
Therefore, they extend uniquely to states $\rho_{k,\infty}$ on the infinite spin chain
algebra $\B(\hil)^{\otimes\infty}$ such that $\rho_{k,n}$ is the $n$-site restriction of
$\rho_{k,\infty}$. Hence, the above hypothesis testing problem can also be considered as
discriminating the global states $\rho_{0,\infty}$ and  $\rho_{1,\infty}$ with local
measurements on an increasing number of sites. Obviously, the unrestricted
i.i.d.~discrimination problem corresponds to $G=\{e\}$ being the trivial group. In this
case we will omit the subscript $G$ from the notations for the error exponents
\eqref{chernoff1}--\eqref{stein3}.

\section{Asymptotic distance measures}\label{sec:distances}
\setcounter{equation}{0}

Let $u$ be a unitary representation of a compact group $G$ on $\hil$ as given in
Section 2, and let $\rho_0$ and $\rho_1$ be two states on $\B(\hil)$ with the density
matrices $\hat\rho_0$ and $\hat\rho_1$. We consider the sequences of states
$\{\rho_{0,n}\}_{n\in\N}$ and $\{\rho_{1,n}\}_{n\in\N}$ as defined in \eqref{F-2.5}.
Note that the densities $\hat\rho_{k,n}$ of $\rho_{k,n}$ with respect to $\Tr$ are given as
$\hat\rho_{k,n}=E_{\cA_n}(\hat\rho_k^{\otimes n})$, $k=0,1$. Define
$$
\psi_n(s):=\log\Tr\hat\rho_{0,n}^s\hat\rho_{1,n}^{1-s},\qquad s\in\bR,\ n\in\bN,
$$
where $\hat\rho_{0,n}^s$ and $\hat\rho_{1,n}^{1-s}$ are defined for all $s\in\bR$ with
convention $0^s=0$ for all $s\in\bR$. In particular, we write
$\supp\rho_{1,n}:=\D{\rho}_{1,n}^0$ for the support projection of $\D{\rho}_{1,n}$.
Also, we define the $\psi$-function in the unrestricted setting as
$$
\psin(s):=\log\Tr\hat\rho_0^s\hat\rho_1^{1-s},\qquad s\in\bR,\ n\in\bN.
$$
Furthermore, let
\begin{equation}\label{F-3.1}
\psi(s):=\lim_{n\to\infty}\frac{1}{n}\psi_n(s)
\end{equation}
whenever the limit exists. Note that $\psi_n$ is finite and convex on $\bR$ as long as
$\rho_{0,n}$ and $\rho_{1,n}$ (more precisely, their supports) are not orthogonal
(otherwise, $\psi_n$ is identically $-\infty$). Hence $\psi$ is convex on any interval
where it exists with values in $[-\infty,+\infty)$. Also, note that if $\rho_{0,n}$ and
$\rho_{1,n}$ are orthogonal, then the same holds for any $m\ge n$.

\begin{lemma}\label{L-3.1}
(1) The sequence $\psi_n(s)$, $n\in\bN$, is subadditive for any $s\in[0,1]$. Hence the
limit \eqref{F-3.1} exists and
\begin{equation*}
\psi(s)=\lim_{n\to\infty}\frac{1}{n}\psi_n(s)
=\inf_{n\ge 1}\frac{1}{n}\psi_n(s),\ds\ds s\in[0,1].
\end{equation*}
Furthermore, if $\rho_0$ and $\rho_1$ are not orthogonal, then $\psi(s)$ is finite
with $\psin(s)\le\psi(s)\le\psi_1(s)$ for all $s\in[0,1]$.

(2) Assume that $\supp\rho_1$ is $G$-invariant (i.e., $\D{\rho}_1^0\in\cA_1$), or that
$\rho_{1,n}$ is faithful for all $n$. Then the sequence $\psi_n(s)$, $n\in\bN$, is
superadditive for any $s\in[1,2]$. Hence the limit \eqref{F-3.1} exists and
\begin{equation*}
\psi(s)=\lim_{n\to\infty}\frac{1}{n}\psi_n(s)
=\sup_{n\ge 1}\frac{1}{n}\psi_n(s),\ds\ds s\in [1,2].
\end{equation*}
Furthermore, if $\supp\rho_1$ is $G$-invariant and $\rho_0$ and $\rho_1$ are not
orthogonal, then $\psi(s)$ is finite with $\psin(s)\ge\psi(s)\ge\psi_1(s)$ for all
$s\in[1,2]$.
\end{lemma}

\begin{proof}
(1) Let $0\le s\le1$. By Lieb's concavity theorem \cite{Li} (see also the Appendix A.1), the
function $(A,B)\mapsto\Tr A^sB^{1-s}$ is jointly concave on the set of
$(A,B)\in\cB(\cH)^{\otimes n+m}\times\cB(\cH)^{\otimes n+m}$, $A,B\ge0$. Note that the
conditional expectation $E_{\cA_n\otimes\cA_m}=E_{\cA_n}\otimes E_{\cA_m}$ is the average
of $\mathrm{Ad}(u_g^{\otimes n}\otimes u_{g'}^{\otimes m})$ by the measure $dg\otimes dg'$
on $G\times G$ (see \eqref{F-2.1}). Since
\begin{align*}
&\Tr\hat\rho_{0,n+m}^s\hat\rho_{1,n+m}^{1-s} \\
&\quad=\Tr\bigl((u_g^{\otimes n}\otimes u_{g'}^{\otimes m})\hat\rho_{0,n+m}
(u_g^{*\otimes n}\otimes u_{g'}^{*\otimes m})\bigr)^s
\bigl((u_g^{\otimes n}\otimes u_{g'}^{\otimes m})\hat\rho_{1,n+m}
(u_g^{*\otimes n}\otimes u_{g'}^{*\otimes m})\bigr)^{1-s}
\end{align*}
for all $g,g'\in G$, Lieb's concavity implies that
\begin{align*}
\Tr\hat\rho_{0,n+m}^s\hat\rho_{1,n+m}^{1-s}
&\le\Tr\bigl(E_{\cA_n\otimes\cA_m}(\hat\rho_{0,n+m})\bigr)^s
\bigl(E_{\cA_n\otimes\cA_m}(\hat\rho_{1,n+m})\bigr)^{1-s} \\
&=\Tr(\hat\rho_{0,n}\otimes\hat\rho_{0,m})^s
(\hat\rho_{1,n}\otimes\hat\rho_{1,m})^{1-s} \\
&=(\Tr\hat\rho_{0,n}^s\hat\rho_{1,n}^{1-s})(\Tr\hat\rho_{0,m}^s\hat\rho_{1.m}^{1-s})
\end{align*}
so that $\psi_{n+m}(s)\le\psi_n(s)+\psi_m(s)$. Furthermore, by the same argument as
above, we have
\begin{equation*}
\bigl(\Tr\hat\rho_0^s\hat\rho_1^{1-s}\bigr)^n
=\Tr\bigl(\hat\rho_0^{\otimes n}\bigr)^s\bigl(\hat\rho_1^{\otimes n}\bigr)^{1-s}
\le\Tr\bigl(E_{\cA_n}(\hat\rho_0^{\otimes n})\bigr)^s
\bigl(E_{\cA_n}(\hat\rho_1^{\otimes n})\bigr)^{1-s}
=\Tr\hat\rho_{0,n}^s\hat\rho_{1,n}^{1-s},
\end{equation*}
and hence
\begin{equation}\label{F-3.2}
\psin(s)\le\inf_{n\ge 1}\frac{1}{n}\psi_n(s)=\psi(s),\qquad s\in[0,1].
\end{equation}
In particular, $\psin(s)\le\psi(s)\le\psi_1(s)$ so that $\psi(s)$ is finite if $\rho_0$
and $\rho_1$ are not orthogonal.

(2) Let $1\le s\le2$ and assume that $\supp\rho_1$ is $G$-invariant. By functional
calculus,
$$
\bz u_g^{\otimes n}\D{\rho}_1^{\otimes n}u_g^{*\otimes n}\jz^0
=((u_g\D{\rho}_1u_{g}^*)^0)^{\otimes n}=(u_g\D{\rho}_1^0u_{g}^*)^{\otimes n}
=(\D{\rho}_1^0)^{\otimes n},
$$
which yields $\D{\rho}_{1,n}^0=(\D{\rho}_1^0)^{\otimes n}$ for all $n\in\bN$. By the same
argument, $((u_g^{\otimes n}\otimes u_{g'}^{\otimes m})\hat\rho_{1,n+m}
(u_g^{*\otimes n}\otimes u_{g'}^{*\otimes m}))^0=(\D{\rho}_1^0)^{\otimes(n+m)}$ for all
$g,g'\in G$ and $n,m\in\N$. In particular, the support of
$\rho_{1,n+m}\circ\Ad(u_g^{\otimes n}\otimes u_{g'}^{\otimes m})$ is the same for all
$g,g'\in G$. This holds trivially also if $\rho_{1,n}$ is faithful for all $n$. Now,
the proof is similar to the above by applying Lemma \ref{L-A.1} of the Appendix instead of
Lieb's theorem. The proof of the remaining part is also similar.
\end{proof}

\begin{cor}\label{C-3.2}
If $\mathrm{supp}\,\rho_0\le\mathrm{supp}\,\rho_1$ then $\psi$ is left-continuous at $1$
as
$$
\lim_{s\nearrow1}\psi(s)=\psi(1)=0.
$$
Similarly, $\mathrm{supp}\,\rho_0\ge\mathrm{supp}\,\rho_1$ implies the right continuity
of $\psi$ at $0$. If $\supp\rho_1$ is $G$-invariant and $\rho_0,\rho_1$  are not
orthogonal (in particular, if  $\rho_1$ is faithful), then $\psi$ is continuous at $1$.
\end{cor}

\begin{proof}
Assume that $\mathrm{supp}\,\rho_0\le\mathrm{supp}\,\rho_1$. Then $\psi_n(1)=0$ for all
$n\in\bN$, and hence $\psi(1)=0$. By \eqref{F-3.2} and the convexity of $\psi$,
\begin{equation*}
0=\psin(1)=\lim_{s\nearrow 1}\psin(s)
\le\lim_{s\nearrow 1}\psi(s)\le\psi(1)=0,
\end{equation*}
and hence $\psi$ is left-continuous at $1$. The proof of the second assertions is similar.
Assume the conditions in the last assertion. Then $\psi$ is a finite-valued convex
function on $[0,2]$ by Lemma \ref{L-3.1}, so that the continuity at $1$ is obvious.
\end{proof}

The \ki{R\'enyi relative entropy} of order $\alpha\in \R\setminus{\{1\}}$ of
$\rho_{0,n}$ with respect to $\rho_{1,n}$ is defined as
\begin{equation*}
S_\alpha(\rho_{0,n}\,\|\,\rho_{1,n})
:=\frac{1}{\alpha-1}\log\Tr \D{\rho}_{0,n}^\alpha\D{\rho}_{1,n}^{1-\alpha}
=-\frac{1}{1-\alpha}\psi_n(\alpha).
\end{equation*}
By Lemma \ref{L-3.1}, the R\'enyi relative entropies with parameter between $0$ and $1$
are superadditive:
\begin{equation}\label{F-3.3}
S_\alpha(\rho_{0,n+m}\,\|\,\rho_{1,n+m})
\ge S_\alpha(\rho_{0,n}\,\|\,\rho_{1,n})+S_\alpha(\rho_{0,m}\,\|\,\rho_{1,m}),
\ds\ds n,m\in\N,\ \alpha\in [0,1),
\end{equation}
and the \ki{mean R\'enyi relative entropy} of order $\alpha$ exists for any
$\alpha\in[0,1)$:
\begin{equation}\label{F-3.4}
S_{\alpha,M}(\rho_0\,\|\,\rho_1)
:=\lim_{n\to\infty}\frac{1}{n}S_\alpha(\rho_{0,n}\,\|\,\rho_{1,n})
=\sup_{n\ge 1}\frac{1}{n}S_\alpha(\rho_{0,n}\,\|\,\rho_{1,n}).
\end{equation}
Similarly, if $\supp\rho_1$ is $G$-invariant or $\rho_{1,n}$ is faithful for all $n$,
then \eqref{F-3.3} and \eqref{F-3.4} hold for the R\'enyi relative entropies with
parameter $\alpha\in (1,2]$. One can easily see that 
\begin{equation*}
\lim_{\alpha\nearrow 1}S_\alpha(\rho_{0,n}\,\|\,\rho_{1,n})
=\sr{\rho_{0,n}}{\rho_{1,n}},
\end{equation*}
where $\sr{\rho_{0,n}}{\rho_{1,n}}$ is the \ki{relative entropy} of $\rho_{0,n}$
with respect to $\rho_{1,n}$, defined as
$$
\sr{\rho_{0,n}}{\rho_{1,n}}
:=\begin{cases}
\Tr\hat\rho_{0,n}(\log\hat\rho_{0,n}-\log\hat\rho_{1,n})
& \text {if $\mathrm{supp}\,\rho_{0,n}\le\mathrm{supp}\,\rho_{1,n}$}, \\
+\infty & \text{otherwise}.
\end{cases}
$$
Hence the relative entropy is also superadditive:
\begin{equation}\label{F-3.5}
S(\rho_{0,n+m}\,\|\,\rho_{1,n+m})\ge
\sr{\rho_{0,n}}{\rho_{1,n}}+S(\rho_{0,m}\,\|\,\rho_{1,m}),
\end{equation}
and the \ki{mean relative entropy} is given by 
\begin{equation*}
\msr{\rho_0}{\rho_1}:=\lim_{n\to\infty}\frac{1}{n}\sr{\rho_{0,n}}{\rho_{1,n}}
=\sup_{n\ge 1}\frac{1}{n}\sr{\rho_{0,n}}{\rho_{1,n}}.
\end{equation*}
Note that the superadditivity \eqref{F-3.5} can also be shown by the monotonicity of the
relative entropy.

\begin{remark}\label{R-3.3}\rm
If we choose the maximally mixed state for $\rho_1$ (i.e., $\hat\rho_1=d^{-1}I_d$), then
\eqref{F-3.3} gives the subadditivity
$$
S_\alpha(\rho_{0,n+m})\le S_\alpha(\rho_{0,n})+S_\alpha(\rho_{0,m}),
\qquad n,m\in\N,\ \alpha\in [0,1),
$$
where
$$
S_\alpha(\rho):={1\over1-\alpha}\log\Tr \hat\rho^\alpha
$$
is the \ki{R\'enyi entropy} of order $\alpha$. This is of some interest since, as is well
known \cite[Chap.~IX, \S6]{Re}, the R\'enyi entropy of order $\alpha$ is not subadditive in general except for the cases
$\alpha=0$ and $\alpha=1$ ($S_1$ denotes the von Neumann entropy).
\end{remark}

We define the \ki{Chernoff distance} of $\rho_{0,n}$ and $\rho_{1,n}$ as 
\begin{equation*}
\chdist{\rho_{0,n}}{\rho_{1,n}}:=-\min_{0\le s\le1}\psi_n(s),
\end{equation*}
and their \ki{Hoeffding distance} with parameter $r\ge 0$ as 
\begin{equation*}
H(r\,|\,\rho_{0,n}\,\|\,\rho_{1,n})
:=\sup_{0\le t<1}\frac{-tr-\psi_n(t)}{1-t}
=\sup_{0\le t<1}\left\{S_t(\rho_{0,n}\,\|\,\rho_{1,n})-\frac{tr}{1-t}\right\}.
\end{equation*}
The mean versions of the above quantities are defined as 
\begin{align}
\chmdist{\rho_0}{\rho_1}
&:=\lim_{n\to\infty}\frac{1}{n}\chdist{\rho_{0,n}}{\rho_{1,n}}, \nonumber\\
\hmdist{\rho_0}{\rho_1}{r}
&:=\lim_{n\to\infty}\frac{1}{n}\hdist{\rho_{0,n}}{\rho_{1,n}}{nr}, \label{F-3.6}
\end{align}
if the limits exist.

We also define the {\it Legendre-Fenchel transforms} (or the polar functions)
\begin{equation}\label{F-3.7}
\ffi_n(a):=\max_{0\le s\le 1}\{as-\psi_n(s)\}
\quad\mbox{and}\quad
\ffi(a):=\sup_{0\le s\le 1}\{as-\psi(s)\},\quad a\in\bR.
\end{equation}
Note that
\begin{equation*}
\chdist{\rho_{0,n}}{\rho_{1,n}}=\ffi_n(0).
\end{equation*}

\begin{lemma}\label{L-3.4}
Let $f_n$, $n\in\bN$, be a superadditive sequence of functions on an interval $I$. Let
$f(s):=\lim_n\frac{1}{n}f_n(s)=\sup_{n}\frac{1}{n}f_n(s)$, the existence of which is
guaranteed by superadditivity. Then
\begin{equation*}
\sup_{s\in I}f(s)=\lim_{n\to\infty}\frac{1}{n}\sup_{s\in I}f_n(s)
=\sup_n\frac{1}{n}\sup_{s\in I}f_n(s).
\end{equation*}
\end{lemma}

\begin{proof}
Obviously, 
\begin{equation*}
\frac{1}{n}\sup_{s\in I} f_n(s)\ge\frac{1}{n}f_n(t),
\qquad n\in\bN,\ t\in I,
\end{equation*}
and thus
\begin{equation*}
\liminf_{n\to\infty}\frac{1}{n}\sup_{s\in I}f_n(s)
\ge\lim_{n\to\infty}\frac{1}{n} f_n(t)=f(t),
\qquad t\in I,
\end{equation*}
which yields
\begin{equation*}
\liminf_{n\to\infty}\frac{1}{n}\sup_{s\in I}f_n(s)\ge\sup_{s\in I}f(s).
\end{equation*}
On the other hand, we have 
\begin{equation*}
\sup_{s\in I}f(s)\ge f(t)\ge\frac{1}{n}f_n(t),\qquad n\in\bN,\ t\in I,
\end{equation*}
and hence
\begin{equation*}
\sup_{s\in I}f(s)\ge \sup_{n\ge 1}\frac{1}{n}\sup_{s\in I}f_n(s),
\end{equation*}
implying the assertion. 
\end{proof}

\begin{prop}\label{P-3.5}
The sequence $\frac{1}{n}\ffi_n(na)$ converges for any $a\in\bR$, and
\begin{equation}\label{F-3.8}
\ffi(a)=\lim_{n\to\infty}\frac{1}{n}\ffi_n(na)=\sup_n\frac{1}{n}\ffi_n(na),
\qquad a\in\bR.
\end{equation}
Moreover, the mean Chernoff bound and the mean Hoeffding bound exist, and
\begin{align}
\chmdist{\rho_0}{\rho_1}&=\sup_{n}\frac{1}{n}\chdist{\rho_{0,n}}{\rho_{1,n}}
=\ffi(0)=-\inf_{0\le s\le 1}\psi(s), \label{F-3.9}\\
\hmdist{\rho_0}{\rho_1}{r}&=\sup_{n}\frac{1}{n}H(nr\,|\,\rho_{0,n}\,\|\,\rho_{1,n})
=\sup_{0\le s<1}\frac{-sr-\psi(s)}{1-s},\qquad r\ge0. \label{F-3.10}
\end{align}
\end{prop}

\begin{proof}
For $0\le s\le1$, the sequence $\psi_n(s)$, $n\in\bN$, is subadditive by
Lemma \ref{L-3.1}. This implies the superadditivity of the sequences
\begin{align*}
f_n(s)&:=nas-\psi_n(s),\qquad0\le s\le1, \\
g_n(s)&:=\frac{-snr-\psi_n(s)}{1-s},\qquad0\le s<1,
\end{align*}
and Lemma \ref{L-3.4} can be applied to obtain \eqref{F-3.8} and \eqref{F-3.10}.
When $a=0$, \eqref{F-3.8} means \eqref{F-3.9}.
\end{proof}

\begin{lemma}\label{hoeffding relentr}
Assume that $\supp\rho_{0,n}\le \supp\rho_{1,n}$ for all $n\in\N$. Then
\begin{equation*}
\hmdist{\rho_0}{\rho_1}{0}=\derleft{\psi}(1)=\msr{\rho_0}{\rho_1},
\end{equation*}
where $\derleft{\psi}(1)$ is 
the left derivative of $\psi$ at $1$.
\end{lemma}
\begin{proof}
The assumption on the supports yields $\psi_n(1)=0$ for all $n\in\N$ and by the convexity
of the $\psi_n$, the functions $s\mapsto \psi_n(t)/(t-1)$ are monotonically increasing,
and hence,
\begin{equation*}
\hdist{\rho_{0,n}}{\rho_{1,n}}{0}=\sup_{0\le t<1}\frac{\psi_n(t)}{t-1}
=\lim_{t\nearrow 1}\frac{\psi_n(t)}{t-1}=\derleft{\psi_n}(1)
=\sr{\rho_{0,n}}{\rho_{1,n}}.
\end{equation*}
Since this holds for all $n\in\N$, we have 
\begin{equation*}
\hmdist{\rho_0}{\rho_1}{0}=\msr{\rho_0}{\rho_1},
\end{equation*}
and $\hmdist{\rho_0}{\rho_1}{0}=\derleft{\psi}(1)$ follows again from $\psi(1)=0$ and the
convexity of $\psi$.
\end{proof}

\begin{remark}\label{R-3.7}\rm
Note that if $\supp\rho_{0,n}\le \supp\rho_{1,n}$ does not hold for some $n$ then it
does not hold for any $m>n$, either. Indeed, if
$\supp\rho_{0,m}\le\supp\rho_{1,m}$ then $\rho_{0,m}\le c\rho_{1,m}$ for some $c>0$,
which implies that $\rho_{0,n}\le c\rho_{1,n}$ for all $n<m$ since
$\rho_{k,m}|_{\cA_n}=\rho_{k,n}$, $k=0,1$.
\end{remark}

Let us also define the Legendre-Fenchel transforms
\begin{equation*}
\tilde\ffi_n(a):=\max_{1\le s\le 3/2}\{a(s-1)-\psi_n(s)\}
\quad\mbox{and}\quad
\tilde\ffi(a):=\sup_{1\le s\le 3/2}\{a(s-1)-\psi(s)\},\quad a\in\bR.
\end{equation*}

\begin{lemma}\label{lemma:strong converse}
If $\supp\rho_1$ is $G$-invariant and $\rho_0,\rho_1$ are not orthogonal (in
particular, if $\rho_1$ is faithful), then the sequence $\frac{1}{n}\tilde\vfi_n(na)$
converges for any $a\in\R$, and
\begin{equation*}
\lim_{n\to\infty}\frac{1}{n}\tilde\vfi_n(na)=\tilde\vfi(a),\qquad a\in\R.
\end{equation*}
\end{lemma}

\begin{proof}
Lemma \ref{L-3.1} implies that $\frac{1}{n}\psi_n(s)$ converges to $\psi(s)$ for every
$s\in[0,2]$ and $\psi$ is a finite-valued convex function on $[0,2]$. Hence the
convergence is uniform on $[1,3/2]$, an interval inside $(0,2)$. Thus,
\begin{align*}
\lim_{n\to\infty}\frac{1}{n}\tilde\vfi_n(na)
&=\lim_{n\to\infty}\max_{1\le s\le 3/2}
\biggl\{a(s-1)-{1\over n}\psi_n(s)\biggr\} \\
&=\max_{1\le s\le 3/2}\lim_{n\to\infty}
\biggl\{a(s-1)-{1\over n}\psi_n(s)\biggr\}=\vfi(a).\qedhere
\end{align*}
\end{proof}

We close this section with the following:

\begin{remark}\label{R-3.9}\rm
Note that $\rho_{1,n}$ is unchanged if the alternative hypothesis $\rho_1$ is replaced by
$\rho_1\circ\Ad u_g$ for any $g\in G$. Hence, by \eqref{F-3.2},
$$
\psin(s\,|\,\rho_0\,\|\,\rho_1\circ\Ad u_g)\le\psi(s),\qquad
s\in[0,1],\ g\in G,
$$
where the above left-hand side denotes $\psin(s)$ for the hypotheses $\rho_0$ and
$\rho_1\circ\Ad u_g$. From \eqref{F-3.9} and \eqref{F-3.10} one obtains
\begin{align}
\chmdist{\rho_0}{\rho_1}
&\le\inf_{g\in G}\chdist{\rho_0}{\rho_1\circ\Ad u_g}
\le\chdist{\rho_0}{\rho_1}, \label{chineq}\\
\hmdist{\rho_0}{\rho_1}{r}
&\le\inf_{g\in G} \hdist{\rho_0}{\rho_1\circ\Ad u_g}{r}
\le\hdist{\rho_0}{\rho_1}{r},\quad r\ge0, \label{hineq}
\end{align}
where
$$
C(\rho_0,\rho_1):=-\min_{0\le s\le1}\psin(s)\ds\ds\text{and}\ds\ds
H(r\,|\,\rho_0\,\|\,\rho_1):=\sup_{0\le s\le1}{-sr-\psin(s)\over1-s}
$$
are the Chernoff and the Hoeffding distances in the unrestricted setting. Also, by the monotonicity of the relative entropy,
\begin{equation}\label{steinineq}
\msr{\rho_0}{\rho_1}
\le\inf_{g\in G} \sr{\rho_0}{\rho_1\circ\Ad u_g}
\le\sr{\rho_0}{\rho_1}.
\end{equation}
Similar inequalities are valid for the error exponents \eqref{chernoff1}--\eqref{stein3}
as well. For example,
$$
\hls{\rho_0}{\rho_1}{r}
\ge\sup_{g\in G}\overline h(r\,|\,\rho_0\,\|\,\rho_1\circ\mathrm{Ad}\,u_g)
\ge\overline h(r\,|\,\rho_0\,\|\,\rho_1),\quad r\ge0.
$$
\end{remark}

\section{Asymptotic error probabilities}\label{sec:error bounds}
\setcounter{equation}{0}

For each $a\in\R$, we define the corresponding minimal asymmetric error probability 
for the discrimination between $\rho_{1,n}$ and $\rho_{1,n}$ as
$$
P_{\min}(a\,|\,\rho_{0,n}:\rho_{1,n}):=\min_{T_n\in\cA_n,\,0\le T_n\le I}
\{e^{-na}\beta_{0,n}(T_n)+\beta_{1,n}(T_n)\},
$$
where $\beta_{0,n}(T_n):=\rho_{0,n}(I-T_n)$ and $\beta_{1,n}(T_n):=\rho_{1,n}(T_n)$ are
the error probabilities of the first and the second kinds for a test $T_n$.
One can easily see that 
\begin{align}
P_{\min}(a\,|\,\rho_{0,n}:\rho_{1,n})
&=e^{-na}\beta_{0,n}(S_{n,a})+\beta_{1,n}(S_{n,a}) \nonumber\\
&={1+e^{-na}\over2}-{1\over2}\|e^{-na}\hat\rho_{0,n}-\hat\rho_{1,n}\|_1, \label{F-4.1}
\end{align}
where
$S_{n,a}:=\{e^{-na}\D{\rho}_{0,n}-\D{\rho}_{1,n}>0\}$ is the spectral projection of the
self-adjoint operator $e^{-na}\D{\rho}_{0,n}-\D{\rho}_{1,n}$ corresponding to the positive
part of its spectrum. $S_{n,a}$ is called a \ki{Neyman-Pearson test} or
\ki{Holevo-Helstr\"om test}. We define the minimal symmetric error probabilities as
$P_{\min}(\rho_{0,n}:\rho_{1,n}):=P_{\min}(0\,|\,\rho_{0,n}:\rho_{1,n})$. 
One can easily see that
\begin{align}
\liminf_{n\to\infty}{1\over n}\log P_{\min}(\rho_{0,n}:\rho_{1,n})
&=\cli{\rho_0}{\rho_1}, \label{F-4.2} \\
\limsup_{n\to\infty}{1\over n}\log P_{\min}(\rho_{0,n}:\rho_{1,n})
&=\cls{\rho_0}{\rho_1}. \label{F-4.3}
\end{align}
The results of \cite{Aud,NSz} on the Chernoff bound says that in the unrestricted case we
have
\begin{equation}\label{F-4.4}
\lim_{n\to\infty}{1\over n}\log P_{\min}(\rho_0^{\otimes n}:\rho_1^{\otimes n})
=-C(\rho_0,\rho_1)=\min_{0\le s\le1}\psin(s).
\end{equation}

We start with the following general observation:

\begin{lemma}\label{lemma:pmin}
For any $a\in\bR$,
\begin{align}
\limsup_{n\to\infty}{1\over n}\log P_{\min}(a\,|\,\rho_{0,n}:\rho_{1,n})
&\le-\ffi(a), \label{F-4.5} \\
\liminf_{n\to\infty}{1\over n}\log P_{\min}(a\,|\,\rho_{0,n}:\rho_{1,n})
&\ge\begin{cases}
2\psi(1/2) & \text{if $a\le0$}, \\
2\psi(1/2)-a & \text{if $a>0$}.
\end{cases} \label{F-4.6}
\end{align}
\end{lemma}

\begin{proof}
Let $A$ and $B$ be positive semidefinite operators $A$ and $B$ on the same Hilbert space. 
By \cite[Theorem 1]{Aud}, 
\begin{equation}\label{F-4.7}
\frac{1}{2}\Tr(A+B)-\frac{1}{2}\norm{A-B}_1\le \Tr A^tB^{1-t},\ds\ds\ds\ds
t\in [0,1],
\end{equation}
and by \cite[Theorem 7]{ANSV},
\begin{equation}\label{F-4.8}
\bz\frac{1}{2}\|A-B\|_1\jz^2+(\Tr A^{1/2}B^{1/2})^2\le\bz\frac{1}{2}\Tr(A+B)\jz^2.
\end{equation}
Let $A:=e^{-na}\D{\rho}_{0,n}$ and 
$B:=\D{\rho}_{1,n}$. Then, \eqref{F-4.7} together with \eqref{F-4.1} yields
\eqref{F-4.5}. On the other hand, by \eqref{F-4.8} we have
\begin{equation*}
{1\over2}\|e^{-na}\hat\rho_{0,n}-\hat\rho_{1,n}\|_1
\le\sqrt{\biggl(\frac{1+e^{-na}}{2}\biggr)^2
-e^{-na}\bigl(\Tr\hat\rho_{0,n}^{1/2}\hat\rho_{1,n}^{1/2}\bigr)^2},
\end{equation*}
and hence, by \eqref{F-4.1},
\begin{align*}
P_{\min}(a\,|\,\rho_{0,n}:\rho_{1,n})
&\ge\frac{1+e^{-na}}{2}-\sqrt{\biggl(\frac{1+e^{-na}}{2}\biggr)^2
-e^{-na}\Tr\hat\rho_{0,n}^{1/2}\hat\rho_{1,n}^{1/2}} \\
&=\frac{e^{-na}\bigl(\Tr\hat\rho_{0,n}^{1/2}\hat\rho_{1,n}^{1/2}\bigr)^2}
{\frac{1+e^{-na}}{2}+\sqrt{\bigl(\frac{1+e^{-na}}{2}\bigr)^2
-e^{-na}\bigl(\Tr\hat\rho_{0,n}^{1/2}\hat\rho_{1,n}^{1/2}\bigr)^2}} \\
&\ge\frac{e^{-na}}{1+e^{-na}}
\bigl(\Tr\hat\rho_{0,n}^{1/2}\hat\rho_{1,n}^{1/2}\bigr)^2\,.
\end{align*}
Therefore,
$$
\liminf_{n\to\infty}\frac{1}{n}\log P_{\min}(a\,|\,\rho_{0,n}:\rho_{1,n})
\ge2\lim_{n\to\infty}\frac{1}{n}\log\Tr\hat\rho_{0,n}^{1/2}\hat\rho_{1,n}^{1/2}
+\lim_{n\to\infty}\frac{1}{n}\log\frac{e^{-na}}{1+e^{-na}},
$$
which yields \eqref{F-4.6}.
\end{proof}

Note that $-\psi(1/2)\le\vfi(0)=\chmdist{\rho_0}{\rho_1}$. By taking account of
\eqref{F-4.2} and \eqref{F-4.3}, Lemma \ref{lemma:pmin}
with the choice $a=0$ yields the following:

\begin{prop}\label{P-4.2}
\begin{equation*}
-2\chmdist{\rho_0}{\rho_1}\le \cli{\rho_0}{\rho_1}
\le\cls{\rho_0}{\rho_1}\le-\chmdist{\rho_0}{\rho_1}.
\end{equation*}
\end{prop}

\begin{prop}\label{prop:Hoeffding}
For any $r\ge 0$,
\begin{equation*}
\hls{\rho_0}{\rho_1}{r}\le-\hmdist{\rho_0}{\rho_1}{r}.
\end{equation*}
\end{prop}

\begin{proof}
Note that $\max\{e^{-na}\beta_{0,n}(S_{n,a}),\beta_{1,n}(S_{n,a})\}
\le P_{\min}(a\,|\,\rho_{0,n}:\rho_{1,n})$, and hence \eqref{F-4.5} yields
\begin{align*}
 \limsup_n\frac{1}{n}\log \beta_{0,n}(S_{n,a})&\le -\vfi(a)+a, \\
 \limsup_n\frac{1}{n}\log \beta_{1,n}(S_{n,a})&\le -\vfi(a).
\end{align*}
Therefore,
\begin{equation*}
\hls{\rho_0}{\rho_1}{r}\le -\sup_{a:\,\vfi(a)-a> r}\vfi(a)
=-\sup_{0\le s<1}\frac{-sr-\psi(s)}{1-s},
\end{equation*}
where the last identity was shown, e.g., in the proof of \cite[Theorem 4.8]{HMO}. A
detailed proof is given in Lemma \ref{L-A.2} of the Appendix. Finally, the right-hand
side of the above inequality is equal to  $-\hmdist{\rho_0}{\rho_1}{r}$ by
Proposition \ref{P-3.5}.
\end{proof}

The following theorem gives the solution of Stein's lemma in our setting:

\begin{thm}\label{thm:Stein}
Assume that $\supp\rho_{0,n}\le\supp\rho_{1,n}$ for all $n\in\N$. Then
$$
\stli{\rho_0}{\rho_1}=\stls{\rho_0}{\rho_1}
=\stl{\rho_0}{\rho_1}=-\msr{\rho_0}{\rho_1}.
$$
\end{thm}

\begin{proof}
We have 
\begin{equation*}
-\msr{\rho_0}{\rho_1}\le\stli{\rho_0}{\rho_1}
\le\stls{\rho_0}{\rho_1}\le\stl{\rho_0}{\rho_1},
\end{equation*}
where the first inequality follows, e.g., from \cite[Proposition 5.2]{HMO}, and the rest
are obvious by definitions. Note that $\stls{\rho_0}{\rho_1}\le \hls{\rho_0}{\rho_1}{0}$
by definition. By taking account of Lemma \ref{hoeffding relentr}, Proposition
\ref{prop:Hoeffding} implies that
\begin{equation*}
\stls{\rho_0}{\rho_1}\le \hls{\rho_0}{\rho_1}{0}
\le-\hmdist{\rho_0}{\rho_1}{0}=-\msr{\rho_0}{\rho_1}.
\end{equation*}
By the definition of $\stls{\rho_0}{\rho_1}$, for each $k\in\N$ 
there exists a sequence of tests $T_{n,k},\,n\in\bN$, such that 
$$
\lim_{n\to\infty}\beta_{0,n}(T_{n,k})=0,\qquad
\limsup_{n\to\infty}\frac{1}{n}\log\beta_{1,n}(T_{n,k})
<-\msr{\rho_0}{\rho_1}+\frac{1}{k}.
$$
For each $k$, we can choose an $n_k\in\bN$ such that for every $n\ge n_k$,
$$
\beta_{0,n}(T_{n,k})<{1\over k},\qquad
\frac{1}{n}\log\beta_{1,n}(T_{n,k})<-\msr{\rho_0}{\rho_1}+\frac{1}{k}.
$$
Here we may assume that $n_1<n_2<\dots$, and we define $T_n^*:=T_{n,k}$ if
$n_k\le n<n_{k+1}$, $k\in\bN$. Obviously, for this sequence of tests,
\begin{equation*}
\lim_{n\to\infty}\beta_{0,n}(T_{n}^*)=0,\qquad
\limsup_{n\to\infty}\frac{1}{n}\log\beta_{1,n}(T_{n}^*)\le-\msr{\rho_0}{\rho_1}.
\end{equation*}
On the other hand, $-\msr{\rho_0}{\rho_1}\le \stli{\rho_0}{\rho_1}$ yields
\begin{equation*}
-\msr{\rho_0}{\rho_1}\le\liminf_{n\to\infty}\frac{1}{n}\log\beta_{1,n}(T_{n}^*),
\end{equation*}
and hence,
\begin{equation*}
\lim_{n\to\infty}\frac{1}{n}\log\beta_{1,n}(T_{n}^*)=-\msr{\rho_0}{\rho_1}.
\end{equation*}
This implies that
\begin{equation*}
\stl{\rho_0}{\rho_1}\le -\msr{\rho_0}{\rho_1},
\end{equation*}
which completes the proof.
\end{proof}

The problem of Stein's lemma can also be formulated in a slightly different way that is
not completely equivalent to the above formulation. For each $\eps\in(0,1)$ and
$n\in\bN$, define the quantity
$$
\beta_\eps(\rho_{0,n}\,\|\,\rho_{1,n}):=\min\{\beta_{1,n}(T_n):
T_n\in\cA_n,\,0\le T_n\le I,\,\beta_{0,n}(T_n)\le\eps\}
$$
and
\begin{align*}
\underline s_{\scriptscriptstyle{G},\ep}(\rho_0\,\|\,\rho_1)
&:=\inf_{\{T_n\}}\biggl\{
\liminf_{n\to\infty} \frac{1}{n}\log\beta_n(T_n)\biggm| \alpha_n(T_n)\le\ep\biggr\}, \\
\overline s_{\scriptscriptstyle{G},\ep}(\rho_0\,\|\,\rho_1)
&:=\inf_{\{T_n\}}\biggl\{
\limsup_{n\to\infty} \frac{1}{n}\log\beta_n(T_n)\biggm| \alpha_n(T_n)\le\ep\biggr\}.
\end{align*}
One can easily see that 
$$
\liminf_{n\to\infty}\frac{1}{n}\log\beta_\eps(\rho_{0,n}\,\|\,\rho_{1,n})
=\underline s_{\scriptscriptstyle{G},\ep}(\rho_0\,\|\,\rho_1),\qquad
\limsup_{n\to\infty}\frac{1}{n}\log\beta_\eps(\rho_{0,n}\,\|\,\rho_{1,n})
=\overline s_{\scriptscriptstyle{G},\ep}(\rho_0\,\|\,\rho_1),
$$
and
\begin{equation*}
\sup_{\ep}\underline s_{\scriptscriptstyle{G},\ep}(\rho_0\,\|\,\rho_1)
=\stli{\rho_0}{\rho_1},\ds\ds\ds
\sup_{\ep}\overline s_{\scriptscriptstyle{G},\ep}(\rho_0\,\|\,\rho_1)
=\stls{\rho_0}{\rho_1}.
\end{equation*}
Hence, Theorem \ref{thm:Stein} implies that if
$\supp\rho_{0,n}\le\supp\rho_{1,n}$ for all $n$ then
\begin{equation}\label{F-4.9}
\underline s_{\scriptscriptstyle{G},\ep}(\rho_0\,\|\,\rho_1)
\le \overline s_{\scriptscriptstyle{G},\ep}(\rho_0\,\|\,\rho_1)
\le -\msr{\rho_0}{\rho_1}
\end{equation}
for all $\ep\in (0,1)$.

\begin{thm}\label{thm:epsilon Stein}
If $\supp\rho_1$ is $G$-invariant and $\supp\rho_0\le\supp\rho_1$ (in particular, if
$\rho_1$ is faithful) then 
\begin{equation*}
-\derright{\psi}(1)\le \underline s_{\scriptscriptstyle{G},\ep}(\rho_0\,\|\,\rho_1)
\le \overline s_{\scriptscriptstyle{G},\ep}(\rho_0\,\|\,\rho_1)\le -\derleft{\psi}(1).
\end{equation*}
\end{thm}

\begin{proof}
In exactly the same way as in \cite{ON}, one can show that 
\begin{equation*}
\beta_{1,n}(T_n)\ge e^{-na}\bz 1-\ep-e^{-\tilde\vfi_n(na)}\jz
\end{equation*}
for any test $T_n$ that satisfies $\beta_{0,n}(T_n)\le \ep$. By Lemma
\ref{lemma:strong converse},
$$
\lim_{n\to\infty}{1\over n}\tilde\vfi_n(na)
=\tilde\vfi(a)=\max_{1\le s\le 3/2}\{a(s-1)-\psi(s)\}.
$$
The latter is strictly positive if and only if $a>\derright{\psi}(1)$, and in this case
$\tilde\vfi_n(na)>(n/2)\tilde\vfi(a)$ for every large enough $n$, and hence
$\lim_n\tilde\vfi_n(na)=+\infty$. Hence,
\begin{equation*}
\liminf_{n\to\infty}\frac{1}{n}\log\beta_{1,n}(T_n)
\ge -a+\lim_{n\to\infty}\frac{1}{n}\log\bz 1-\ep-e^{-\tilde\vfi_n(na)}\jz=-a.
\end{equation*}
Since this is true for all $a>\derright{\psi}(1)$, we get
$\underline s_{\scriptscriptstyle{G},\ep}(\rho_0\,\|\,\rho_1)\ge -\derright{\psi}(1)$.
The rest of the inequalities are just a restatement of \eqref{F-4.9} thanks to Lemma
\ref{hoeffding relentr}.
\end{proof}

Theorem \ref{thm:epsilon Stein} together with Lemma \ref{hoeffding relentr} yields
immediately the following:

\begin{cor}\label{C-4.6}
Assume that the conditions of Theorem \ref{thm:epsilon Stein} hold and, moreover, that
$\psi$ is differentiable at $1$. Then
\begin{equation*}
\lim_{n\to\infty}\frac{1}{n}\log\beta_\eps(\rho_{0,n}\,\|\,\rho_{1,n})
=-\msr{\rho_0}{\rho_1}
\end{equation*}
for all $\eps\in(0,1)$.
\end{cor}

We close this section with the following:

\begin{remark}\label{rem:differentiability}\rm
The analysis in \cite{HMO} shows that if $\psi$ exists and is differentiable on the whole
real line and $\derleft{\psi}(1)=\msr{\rho_0}{\rho_1}$
(see Lemma \ref{hoeffding relentr}) then 
\begin{align*}
\cli{\rho_0}{\rho_1}=\cls{\rho_0}{\rho_1}
=\cl{\rho_0}{\rho_1}=&-\chmdist{\rho_0}{\rho_1},\\
\hli{\rho_0}{\rho_1}{r}=\hls{\rho_0}{\rho_1}{r}
=\hl{\rho_0}{\rho_1}{r}=&-\hmdist{\rho_0}{\rho_1}{r},\ds\ds r\ge 0,\\
\stli{\rho_0}{\rho_1}=\stls{\rho_0}{\rho_1}=\stl{\rho_0}{\rho_1}=&-\msr{\rho_0}{\rho_1}.
\end{align*} 
(Actually, it is enough to require the existence and differentiability of $\psi$ on the
open interval $(0,1)$ to show the above identities based on the G\"artner-Ellis
theorem.) Note that \eqref{F-4.2} and \eqref{F-4.3} imply in this case that
\begin{equation}\label{F-4.10}
\lim_{n\to\infty}{1\over n}\log P_{\min}(\rho_{0,n}:\rho_{1,n})=-\chmdist{\rho_0}{\rho_1}.
\end{equation}
In Section \ref{sec:examples} we will show some examples where $\psi$ can be
explicitly computed and shown to be differentiable on $\bR$, and hence the above
identities hold.
\end{remark}

\section{Asymptotic distance measures for an invariant alternative hypothesis}\label{sec:invariant}
\setcounter{equation}{0}

As the examples of Section \ref{sec:examples} will show, the performance of the
$G$-invariant and the unrestricted measurements can be very different in general. In
particular, the states might be perfectly distinguishable by unrestricted measurements,
while completely indistinguishable by $G$-invariant ones. As our following discussion
shows, this cannot happen if the alternative hypothesis is invariant under the symmetry
group. In the first part, we show that in this case $G$-invariant measurements perform
just as well as unrestricted ones in the setting of Stein's lemma. This follows
immediately from Theorem \ref{T-5.1}, thanks to Theorem \ref{thm:Stein}. Although
the same is not true for the settings of the Chernoff and the Hoeffding bounds (see
Example \ref{E-6.2}), it is still possible to establish a strong relation between the
different performances in the setting of the Chernoff bound as is shown the second part
of this section.

\subsection{Mean relative entropy}

We prove the following partial extension of \cite[Theorem 2.1]{HP}, improving the arguments
in \cite{HP} based on \eqref{F-2.4}. In the proof, we use the same notations as defined
in Section 2 for the irreducible decompositions of the tensor powers of the representation
$u$. 

\begin{thm}\label{T-5.1}
If $\rho_1$ is $G$-invariant (i.e., $\hat\rho_1\in\cA_1$), then
\begin{align*}
S(\rho_0\,\|\,\rho_1)&=\msr{\rho_0}{\rho_1} \\
&=\lim_{n\to\infty}{1\over n}\sup\bigl\{S(\rho_{0,n}|_\cB\,\|\,\rho_{1,n}|_\cB):
\mbox{$\cB$ is an abelian subalgebra of $\cA_n$}\bigr\}.
\end{align*}
\end{thm}

\begin{proof}
The monotonicity of the relative entropy implies that 
$$
S(\rho_0^{\otimes n}\,\|\,\rho_1^{\otimes n})
\ge S(\rho_{n,0}\,\|\,\rho_{n,1})
\ge S(\rho_{0,n}|_\cB\,\|\,\rho_{1,n}|_\cB)
$$
for any subalgebra $\B\subset\A_n$ and hence,
\begin{align*}
S(\rho_0\,\|\,\rho_1)&\ge \msr{\rho_0}{\rho_1} \\
&\ge\limsup_{n\to\infty}{1\over n}\sup\bigl\{S(\rho_{0,n}|_\cB\,\|\,\rho_{1,n}|_\cB):
\mbox{$\cB$ is an abelian subalgebra of $\cA_n$}\bigr\}.
\end{align*}
The assumption $\hat\rho_1\in\cA_1$ implies that $\hat\rho_1^{\otimes n}\in\cA_n$ for
all $n\in\bN$. By \eqref{F-2.2},
$$
\hat\rho_{1,n}=\hat\rho_1^{\otimes n}
=\bigoplus_{i=1}^{k_n}\Bigl(D_i^{(n)}\otimes I_{d_i^{(n)}}\Bigr)
\quad\mbox{with}\quad D_i^{(n)}\in M_{m_i^{(n)}}.
$$
With the spectral decomposition
$D_i^{(n)}=\sum_{j=1}^{l_i^{(n)}}\lambda_{ij}^{(n)}P_{ij}^{(n)}$ for $1\le i\le k_n$, we
define
\begin{equation}\label{F-5.1}
\cE_n(A):=\sum_{i=1}^{k_n}\sum_{j=1}^{l_i^{(n)}}
\bigl(P_{ij}^{(n)}\otimes I_{d_i^{(n)}}\bigr)E_{\cA_n}(A)
\bigl(P_{ij}^{(n)}\otimes I_{d_i^{(n)}}),\qquad A\in\B(\hil)^{\otimes n}.
\end{equation}
Let $\cB_n$ denote the abelian subalgebra of $\cA_n$ generated by
$$
\bigl(P_{ij}^{(n)}\otimes I_{d_i^{(n)}}\bigr)\hat\rho_0^{\otimes n}
\bigl(P_{ij}^{(n)}\otimes I_{d_i^{(n)}}\bigr),
\qquad1\le i\le k_n,\ 1\le j\le l_i^{(n)}.
$$
Then, as in the proof of \cite[Lemma 3.1]{HP}, we have
\begin{equation}\label{F-5.2}
nS(\rho_0\,\|\,\rho_1)=S(\rho_0^{\otimes n}\,\|\,\rho_1^{\otimes n})
=S(\rho_0^{\otimes n}|_{\cB_n}\,\|\,\rho_1^{\otimes n}|_{\cB_n})
+S(\rho_0^{\otimes n}\circ\cE_n)-S(\rho_0^{\otimes n}).
\end{equation}

Similarly to \cite[Lemma 3.2]{HP} we next prove that
\begin{equation}\label{F-5.3}
S(\omega\circ\cE_n)-S(\omega)
\le d\log(n+1)+2\log\Biggl(\sum_{i=1}^{k_n} d_i^{(n)}\Biggr)
\end{equation}
for any state $\omega$ on $\B(\hil)^{\otimes n}$.
Note that $S(\omega\circ\cE_n)-S(\omega)=S(\omega\,\|\,\omega\circ\cE_n)$ and hence, by
the joint convexity of the relative entropy, it is enough to show
\eqref{F-5.3} for pure states. Assume thus that $\D{\omega}=\pr{\psi}$ with
some unit vector $\psi\in\hil^{\otimes n}$. By \eqref{F-2.3} we write
\begin{equation*}
E_{\cA_n}(\D{\omega})=\sum_{i=1}^{k_n}E_i^{(n)}(A_i^{(n)})\quad
\mbox{with}\quad A_i^{(n)}:=\pr{P_i^{(n)}\psi}\in M_{m_i^{(n)}}\otimes M_{d_i^{(n)}}
\end{equation*}
and $E_i^{(n)}(A_i^{(n)})=B_i^{(n)}\otimes I_{d_i^{(n)}}$ with
$B_i^{(n)}\in M_{m_i^{(n)}}$. 
The rank of $B_i^{(n)}$ is equal to the Schmidt rank of the vector $P_i^{(n)}\psi$, which
is upper bounded by $d_i^{(n)}$. Hence the rank of $E_i^{(n)}(A_i^{(n)})$ is at most
$(d_i^{(n)})^2$. One can also see that the number of different eigenvalues of
$\D{\rho}_1^{\otimes n}$ is upper bounded by $(n+1)^d$ and hence $l_i^{(n)}\le (n+1)^d$,
$1\le i\le k_n$. Thus, by \eqref{F-5.1}, the rank of $\cE_n(\pr{\psi})$ is upper bounded
by $(n+1)^d\sum_{i=1}^{k_n}(d_i^{(n)})^2$ and therefore,
$S(\omega\circ\cE_n)-S(\omega)=S(\omega\circ\cE_n)$ is dominated by
$$
\log\Biggl((n+1)^d\sum_{i=1}^{k_n}(d_i^{(n)})^2\Biggr)
\le d\log(n+1)+2\log\Biggl(\sum_{i=1}^{k_n} d_i^{(n)}\Biggr).
$$

Now we apply \eqref{F-5.3} to \eqref{F-5.2} with the choice
$\omega:= \rho_0^{\otimes n}$ and use \eqref{F-2.4} to obtain
\begin{equation*}
S(\rho_0\,\|\,\rho_1)
\le\liminf_{n\to\infty}{1\over n}S(\rho_{0,n}|_{\cB_n}\,\|\,\rho_{1,n}|_{\cB_n}),
\end{equation*}
which completes the proof.
\end{proof}

\subsection{Fidelity and Chernoff bound}\label{sec:fidelity}

In this section we will discuss the minimal symmetric error probability based on the
relation between the fidelity and the Chernoff bound in the setting with group symmetry.
The {\it fidelity} of two states $\rho$ and $\sigma$ on a matrix algebra is given by
$$
F(\rho,\sigma):=\Tr\big|\hat\rho^{1/2}\hat\sigma^{1/2}\big|
=\Tr\bigl(\hat\sigma^{1/2}\hat\rho\hat\sigma^{1/2}\bigr)^{1/2}
=\Tr\bigl(\hat\rho^{1/2}\hat\sigma\hat\rho^{1/2}\bigr)^{1/2},
$$
which is used as a distinguishability measure in quantum hypothesis testing. It is known
(see \cite[Theorem 1]{FvdG}, \cite[(9.110)]{NC}) that
\begin{equation}\label{F-5.4}
{1-\sqrt{1-F(\rho,\sigma)^2}\over2}\le P_{\min}(\rho:\sigma)
\le{F(\rho,\sigma)\over2}.
\end{equation}
Furthermore, it is also well known (see \cite[Theorem 9.6]{NC}, \cite[Theorem 6.2]{Pe2})
that the fidelity $F(\rho,\sigma)$ is monotone increasing under trace-preserving
completely positive maps. Hence we have
\begin{equation}\label{F-5.5}
F(\rho_{0,n},\rho_{1,n})\ge F(\rho_0^{\otimes n},\rho_1^{\otimes n})
=F(\rho_0,\rho_1)^n.
\end{equation}

The following inequality was proved in \cite[Theorem 6]{ANSV}. Here, we provide an
alternative proof.

\begin{lemma}\label{L-5.2}
For every states $\rho$ and $\sigma$ on $\cB(\hil)$ and for every $0\le s\le1$,
$$
\Tr\hat\rho^s\hat\sigma^{1-s}\ge F(\rho,\sigma)^2.
$$
\end{lemma}

\begin{proof}
By the fidelity formula with purifications due to Uhlmann (see \cite[Lemma 8.2]{Ha},
\cite[Theorem 9.4]{NC}), there are purifications $|\ffi\>\<\ffi|$ and
$|\psi\>\<\psi|$ of $\rho$ and $\sigma$, respectively, such that
$$
F(\rho,\sigma)=|\<\ffi,\psi\>|.
$$
By Lieb's concavity theorem (see the Appendix A.1 for details), we then
have
\begin{equation*}
\Tr\hat\rho^s\hat\sigma^{1-s}\ge\Tr(|\ffi\>\<\ffi|)^s(|\psi\>\<\psi|)^{1-s}
=|\<\ffi,\psi\>|^2=F(\rho,\sigma)^2.\qedhere
\end{equation*}
\end{proof}

We need only the following \eqref{F-5.7} with $s=1/2$ for later use, but the extended
inequalities are of some interest in themselves.

\begin{thm}\label{T-5.3}
If $\rho_1$ is $G$-invariant (i.e., $\hat\rho_1\in\cA_1$), then
\begin{align}
\liminf_{n\to\infty}{1\over n}\log\Tr|\hat\rho_{0,n}^s\hat\rho_{1,n}^{1-s}|
&\ge\log\Tr|\hat\rho_0^s\hat\rho_1^{1-s}|\quad\mbox{if $0\le s\le1/2$},
\label{F-5.6}\\
\limsup_{n\to\infty}{1\over n}\log\Tr|\hat\rho_{0,n}^s\hat\rho_{1,n}^{1-s}|
&\le\log\Tr|\hat\rho_0^s\hat\rho_1^{1-s}|\quad\mbox{if $1/2\le s\le1$}.
\label{F-5.7}
\end{align}
\end{thm}

\begin{proof}
With the same notations as in Section 2, $\hat\rho_{0,n}$ is written as
$$
\hat\rho_{0,n}=E_{\cA_n}(\hat\rho_0^{\otimes n})
=\sum_{i=1}^{k_n}E_i^{(n)}(P_i^{(n)}\hat\rho_0^{\otimes n}P_i^{(n)}),
$$
(see \eqref{F-2.3}) while $\hat\rho_{1,n}=\hat\rho_1^{\otimes n}\in\cA_n$ by the
assumption $\hat\rho_1\in\cA_1$.

First, we prove \eqref{F-5.7}. Since $1\le2s\le2$, note that $x^{2s}$ is an operator
convex function on $[0,+\infty)$. Hence we have
$\hat\rho_{0,n}^{2s}\le E_{\cA_n}((\hat\rho_0^{\otimes n})^{2s})$ so that
\begin{align}
\Tr|\hat\rho_{0,n}^s\hat\rho_{1,n}^{1-s}|
&=\Tr(\hat\rho_{1,n}^{1-s}\hat\rho_{0,n}^{2s}\hat\rho_{1,n}^{1-s})^{1/2} \nonumber\\
&\le\Tr(E_{\cA_n}(\hat\rho_{1,n}^{1-s}(\hat\rho_0^{\otimes n})^{2s}
\hat\rho_{1,n}^{1-s}))^{1/2} \nonumber\\
&=\Tr\Biggl(\sum_{i=1}^{k_n}E_i^{(n)}(P_i^{(n)}\hat\rho_{1,n}^{1-s}
(\hat\rho_0^{\otimes n})^{2s}\hat\rho_{1,n}^{1-s}P_i^{(n)})\Biggr)^{1/2}. \label{F-5.8}
\end{align}
Set $A_i^{(n)}:=P_i^{(n)}\hat\rho_{1,n}^{1-s}(\hat\rho_0^{\otimes n})^{2s}
\hat\rho_{1,n}^{1-s}P_i^{(n)}$ for $1\le i\le k_n$.
By Lemma \ref{L-A.3} of the Appendix, for $1\le i\le k_n$ there are unitaries
$U_{i,j}^{(n)}\in I_{m_i^{(n)}}\otimes M_{d_i^{(n)}}$, $1\le j\le(d_i^{(n)})^2$, such that
$$
E_i^{(n)}(A_i^{(n)})={1\over (d_i^{(n)})^2}\sum_{j=1}^{(d_i^{(n)})^2}
U_{i,j}^{(n)}A_i^{(n)}U_{i,j}^{(n)*}.
$$
Hence we have
\begin{align}
\Tr\Biggl(\sum_{i=1}^{k_n}E_i^{(n)}(A_i^{(n)})\Biggr)^{1/2}
&\le\Tr\Biggl(\sum_{i=1}^{k_n}\sum_{j=1}^{(d_i^{(n)})^2}
U_{i,j}^{(n)}A_i^{(n)}U_{i,j}^{(n)*}\Biggr)^{1/2} \nonumber\\
&\le\sum_{i=1}^{k_n}\sum_{j=1}^{(d_i^{(n)})^2}
\Tr(U_{i,j}^{(n)}A_i^{(n)}U_{i,j}^{(n)*})^{1/2} \nonumber\\
&=\sum_{i=1}^{k_n}(d_i^{(n)})^2\Tr(A_i^{(n)})^{1/2}. \label{F-5.9}
\end{align}
In the above, the first inequality is just removing $1/(d_i^{(n)})^2$, and
for the second inequality, see \cite[Eq.\,(12)]{BH} (or an extended result in \cite{AZ}).
Combining \eqref{F-5.8} and \eqref{F-5.9} yields
\begin{align*}
\Tr|\hat\rho_{0,n}^s\hat\rho_{1,n}^{1-s}|
&\le\sum_{i=1}^{k_n}(d_i^{(n)})^2\Tr(P_i^{(n)}\hat\rho_{1,n}^{1-s}
(\hat\rho_0^{\otimes n})^{2s}\hat\rho_{1,n}^{1-s}P_i^{(n)})^{1/2} \\
&=\sum_{i=1}^{k_n}(d_i^{(n)})^2\Tr\bigl((\hat\rho_{1,n}^{1-s}
(\rho_0^{\otimes n})^{2s}\hat\rho_{1,n}^{1-s})^{1/2}P_i^{(n)}
(\hat\rho_{1,n}^{1-s}(\rho_0^{\otimes n})^{2s}\hat\rho_{1,n}^{1-s})^{1/2}\bigr)^{1/2} \\
&\le\Biggl(\sum_{i=1}^{k_n}d_i^{(n)}\Biggr)^2
\Tr(\hat\rho_{1,n}^{1-s}(\rho_0^{\otimes n})^{2s}\hat\rho_{1,n}^{1-s})^{1/2} \\
&=\Biggl(\sum_{i=1}^{k_n}d_i^{(n)}\Biggr)^2
\Tr|(\hat\rho_0^{\otimes n})^s(\hat\rho_1^{\otimes n})^{1-s}| \\
&=\Biggl(\sum_{i=1}^{k_n}d_i^{(n)}\Biggr)^2
\bigl(\Tr|\hat\rho_0^s\hat\rho_1^{1-s}|\bigr)^n.
\end{align*}
Thanks to \eqref{F-2.4} we obtain inequality \eqref{F-5.7}.

Next, we prove \eqref{F-5.6}. Since $x^{2s}$ is operator concave on $[0,+\infty)$
thanks to $0\le2s\le1$, we have $\rho_{0,n}^{2s}\ge E_{\cA_n}((\rho_0^{\otimes n})^{2s})$.
Hence inequality \eqref{F-5.8} is reversed. Inequality \eqref{F-5.9} is also reversed as
follows:
\begin{align*}
\Tr\Biggl(\sum_{i=1}^{k_n}E_i^{(n)}(A_i^{(n)})\Biggr)^{1/2}
&=\Tr\sum_{i=1}^{k_n}\bigl(E_i^{(n)}(A_i^{(n)})\bigr)^{1/2} \\
&\ge\Tr\sum_{i=1}^{k_n}E_i^{(n)}\bigl((A_i^{(n)})^{1/2}\bigr) \\
&\ge\Tr\sum_{i=1}^{k_n}P_i^{(n)}(\hat\rho_{1,n}^{1-s}
(\hat\rho_0^{\otimes n})^{2s}\hat\rho_{1,n}^{1-s})^{1/2}P_i^{(n)} \\
&=\bigl(\Tr|\hat\rho_0^s\hat\rho_1^{1-s}|\bigr)^n.
\end{align*}
In the above, the first inequality follows by the operator concavity of the square root
function, and the second one follows from \cite{Han}. Hence we have inequality
\eqref{F-5.6}.
\end{proof}

Theorem \ref{T-5.3} for $s=1/2$ together with \eqref{F-5.5} yields 

\begin{cor}\label{C-5.4}
If $\hat\rho_1\in\cA_1$ then
$$
\lim_{n\to\infty}{1\over n}\log F(\rho_{0,n},\rho_{1,n})
=\inf_{n\ge1}{1\over n}\log F(\rho_{0,n},\rho_{1,n})=\log F(\rho_0,\rho_1).
$$
\end{cor}

Note that the logarithmic fidelity $ -\log F(\cdot\,,\,\cdot)$ is a generalized relative
entropy in the sense that (i) it takes strictly positive values on unequal states and zero
if its arguments are equal, (ii) it is monotonically decreasing under trace-preserving
completely positive maps, and (iii) it is jointly convex in its arguments. In view of this,
Corollary \ref{C-5.4} is a direct analogue of Theorem \ref{T-5.1}. The extremal case in
Example \ref{E-6.1} of the next section shows that assuming the $G$-invariance of $\rho_1$
is essential for Theorem \ref{T-5.3} and Corollary \ref{C-5.4}.

As Example \ref{E-6.2} shows, the $G$-invariance of $\rho_1$ does not imply the same
asymptotics for the restricted and the unrestricted minimal error probabilities. However,
one can still obtain the following non-trivial bound:

\begin{thm}\label{T-5.5}
If $\hat\rho_1\in\cA_1$ then
\begin{align*}
\lim_{n\to\infty}{1\over n}\log P_{\min}(\rho_0^{\otimes n}:\rho_1^{\otimes n})
&\le\liminf_{n\to\infty}{1\over n}\log P_{\min}(\rho_{0,n}:\rho_{1,n}) \\
&\le\limsup_{n\to\infty}{1\over n}\log P_{\min}(\rho_{0,n}:\rho_{1,n}) \\
&\le{1\over2}\lim_{n\to\infty}{1\over n}
\log P_{\min}(\rho_0^{\otimes n}:\rho_1^{\otimes n}).
\end{align*}
\end{thm}

\begin{proof}
By Lemma \ref{L-5.2} and Corollary \ref{C-5.4} we have
\begin{align*}
{1\over2}\lim_{n\to\infty}{1\over n}
\log P_{\min}(\rho_0^{\otimes n}:\rho_1^{\otimes n})
&={1\over2}\min_{0\le s\le1}\log\Tr\hat\rho_0^s\hat\rho_1^{1-s} \\
&\ge\log F(\rho_0,\rho_1) \\
&=\lim_{n\to\infty}{1\over n}\log F(\rho_{0,n},\rho_{1,n}) \\
&\ge\limsup_{n\to\infty}{1\over n}\log P_{\min}(\rho_{0,n}:\rho_{1,n})
\end{align*}
thanks to \eqref{F-5.4}. Hence the last inequality follows, and the others are obvious.
\end{proof}

By \eqref{F-4.2}--\eqref{F-4.4}, the inequalities of the above theorem can be
rewritten as
$$
-\chdist{\rho_0}{\rho_1}\le\cli{\rho_0}{\rho_1}
\le\cls{\rho_0}{\rho_1}\le-{1\over2}\chdist{\rho_0}{\rho_1}.
$$
Comparing these with the inequalities of Proposition \ref{P-4.2} and also taking account
of Remark \ref{rem:differentiability}, we have the following:

\begin{cor}\label{C-5.6}
Assume that $\rho_1$ is $G$-invariant. Then
\begin{equation*}
\frac{1}{4}\chdist{\rho_0}{\rho_1}\le \chmdist{\rho_0}{\rho_1}\le \chdist{\rho_0}{\rho_1}.
\end{equation*}
Moreover, ${1\over2}\chdist{\rho_0}{\rho_1}\le \chmdist{\rho_0}{\rho_1}$ holds whenever
$\psi$ is differentiable on $(0,1)$.
\end{cor}

\begin{remark}\label{R-5.7}\rm
The constant $1/2$ in Theorem \ref{T-5.5} is actually the best possible, as will be
seen in Remark \ref{R-6.4} of the next section. This also shows that in the case of a
differentiable $\psi$, the bound
${1\over2}\chdist{\rho_0}{\rho_1}\le \chmdist{\rho_0}{\rho_1}\le\chdist{\rho_0}{\rho_1}$
is the best possible.
\end{remark}

\section{Restricted vs.~unrestricted measurements: examples}\label{sec:examples}
\setcounter{equation}{0}

In this section, we illustrate, through some examples, the difference between the
performance of $G$-invariant measurements and that of unrestricted ones. As the
following example shows, the difference can be as extreme as possible even in the
classical situation where the densities corresponding to the null and the alternative
hypotheses are commuting. This also shows that the assumption that $\rho_1$ is
$G$-invariant cannot be removed in Theorem \ref{T-5.1}.

\begin{example}\label{E-6.1}\rm{\bf(Two commuting states with $\bZ_2$-symmetry)}\quad
Let $G:=\bZ_2=\{\pm1\}$ and $u$ be the representation of $G$ on $\hil:=\bC^2$ with
$u_{-1}:=\bmatrix1&0\\0&-1\endbmatrix$. Then the $\bZ_2$-fixed point subalgebra of
$\cB(\hil)^{\otimes n}$ is $\cA_n=M_{2^{n-1}}\oplus M_{2^{n-1}}$. Consider a commuting set
of states $\sigma_\lambda$, $0\le\lambda\le1$, given by
$$
\hat\sigma_\lambda:=\lambda\,|+\>\<+|+(1-\lambda)\,|-\>\<-|
=\lambda\bmatrix1/2&1/2\\1/2&1/2\endbmatrix
+(1-\lambda)\bmatrix1/2&-1/2\\-1/2&1/2\endbmatrix,
$$
where $|+\rangle:=(1/\sqrt2,1/\sqrt2)$ and $|-\rangle:=(1/\sqrt2,-1/\sqrt2)$. We write
$$
\hat\sigma_\lambda^{\otimes n}
=\sum_{i=0}^n\lambda^i(1-\lambda)^{n-i}E_{n,i}\quad\mbox{with}\quad
E_{n,i}:=\sum_{e_j=\pm,\,\#\{j:e_j=+\}=i}\otimes_{j=1}^{n}\pr{e_j}.
$$
Since $u_{-1}^{\otimes n}E_{n,i}u_{-1}^{*\otimes n}=E_{n,n-i}$, we have
\begin{equation}\label{F-6.1}
E_{\cA_n}(\hat\sigma_\lambda^{\otimes n})
={1\over2}(\hat\sigma_\lambda^{\otimes n}
+u_{-1}^{\otimes n}\hat\sigma_\lambda^{\otimes n}u_{-1}^{*\otimes n}) 
=\sum_{i=0}^n{\lambda^i(1-\lambda)^{n-i}+\lambda^{n-i}(1-\lambda)^i\over2}\,E_{n,i}.
\end{equation}

Now let $\rho_0:=\sigma_\lambda$ and $\rho_1:=\sigma_\mu$ with any $\lambda,\mu\in[0,1]$.
When $0\le s\le1$, noting that $2^{s-1}(a^s+b^s)\le(a+b)^s\le a^s+b^s$ for all $a,b\ge0$
and that $\Tr E_{n,i}={n\choose i}$, we see from \eqref{F-6.1} that
$\Tr \rho_{0,n}^s\rho_{1,n}^{1-s}$ is upper bounded by
\begin{align}
&{1\over2}\sum_{i=0}^n
\bigl\{(\lambda^i(1-\lambda)^{n-i})^s+(\lambda^{n-i}(1-\lambda)^i)^s\bigr\}
\bigl\{(\mu^i(1-\mu)^{n-i})^{1-s}+(\mu^{n-i}(1-\mu))^{1-s}\bigr\}{n\choose i} \nonumber\\
&\quad=\sum_{i=0}^n\bigl\{(\lambda^s\mu^{1-s})^i((1-\lambda)^s(1-\mu)^{1-s})^{n-i}
+(\lambda^s(1-\mu)^{1-s})^i((1-\lambda)^s\mu^{1-s})^{n-i}\bigr\}{n\choose i} \nonumber\\
&\quad=(\lambda^s\mu^{1-s}+(1-\lambda)^s(1-\mu)^{1-s})^n
+(\lambda^s(1-\mu)^{1-s}+(1-\lambda)^s\mu^{1-s})^n \label{F-6.2}
\end{align}
and also lower bounded by $1/2$ times \eqref{F-6.2}. Therefore,
\begin{align*}
\psi(s)&=\max\bigl\{\log(\lambda^s\mu^{1-s}+(1-\lambda)^s(1-\mu)^{1-s}),
\log(\lambda^s(1-\mu)^{1-s}+(1-\lambda)^s\mu^{1-s})\bigr\} \\
&=\max\bigl\{\psin(s\,|\,\sigma_\lambda\,\|\,\sigma_\mu),
\psin(s\,|\,\sigma_{\lambda}\,\|\,\sigma_{1-\mu})\bigr\},
\qquad0\le s\le1,
\end{align*}
where $\psin(s\,|\,\sigma_\lambda\,\|\,\sigma_\mu)$ denotes $\psin(s)$ for
$\rho_0=\sigma_{\lambda}$ and $\rho_1=\sigma_{\mu}$. Note that 
\begin{equation}\label{cases of psi}
\psi(s)=\begin{cases}
\psin(s\,|\,\sigma_\lambda\,\|\,\sigma_\mu) & \text{if $(1/2-\lambda)(1/2-\mu)\ge0$},\\
\psin(s\,|\,\sigma_\lambda\,\|\,\sigma_{1-\mu})& \text{if $(1/2-\lambda)(1/2-\mu)<0$}.
\end{cases}
\end{equation}
In particular, $\psi$ is differentiable on $(0,1)$ except when $\lambda=0$ and $\mu=1$ or
the other way around.

Note that
$\sigma_{1-\mu}=\sigma_\mu\circ\mathrm{Ad}\,u_{-1}$, and \eqref{cases of psi} yields
\begin{align*}
\chmdist{\rho_0}{\rho_1}&=\min\{\chdist{\sigma_\lambda}{\sigma_\mu},
\chdist{\sigma_{\lambda}}{\sigma_{1-\mu}}\}, \\
\hmdist{\rho_0}{\rho_1}{r}&=\min\{\hdist{\sigma_\lambda}{\sigma_\mu}{r},
\hdist{\sigma_\lambda}{\sigma_\mu}{r}\},\quad r\ge0, \\
\msr{\rho_0}{\rho_1}&=\min\{\sr{\sigma_\lambda}{\sigma_\mu},
\sr{\sigma_{\lambda}}{\sigma_{1-\mu}}\},
\end{align*}
and hence the first inequality in each of \eqref{chineq}--\eqref{steinineq} hold with
equality. On the other hand, \eqref{cases of psi} shows that if $(1/2-\lambda)(1/2-\mu)<0$
then
$$
\psi(s)=\psin(s\,|\,\sigma_\lambda\,\|\,\sigma_{1-\mu})
>\psin(s\,|\,\sigma_\lambda\,\|\,\sigma_{\mu})=\psin(s)
$$
so that for any $r\ge0$,
\begin{equation}\label{strict inequalities}
\chmdist{\rho_0}{\rho_1}<\chdist{\rho_0}{\rho_1},\ds
\hmdist{\rho_0}{\rho_1}{r}<\hdist{\rho_0}{\rho_1}{r},\ds
\msr{\rho_0}{\rho_1}<\sr{\rho_0}{\rho_1}.
\end{equation}
The differentiability of $\psi$ on $(0,1)$ implies that the identities of Remark
\ref{rem:differentiability} hold (as long as $0<\mu<1$), and hence
\eqref{strict inequalities} shows that $G$-invariant measurements perform strictly worse
than unrestricted ones in all of the settings of the Chernoff and Hoeffding bounds and
of Stein's lemma.  In particular, in the extremal case where $\rho_0=\sigma_0$ and
$\rho_1=\sigma_1$, the two states have orthogonal supports and hence unrestricted
measurements yield a perfect distinguishability, while one can easily see that 
$\rho_{0,n}=\rho_{1,n}$ for all $n$ so that the states are completely indistinguishable
with $G$-invariant measurements. Note that in this case $\psi(s)=0$ while
$\psin(s)=-\infty$ for all $s\in\bR$.

Finally, we show for completeness that $\psi(s)$ exists and is differentiable on the whole
real line. Replacing $\lambda$ by $1-\lambda$, $\mu$ by $1-\mu$ and $s$ by $1-s$ if
necessary, we may assume that $0<\lambda<\mu\le1/2$, since the cases $\lambda=0$ and
$\lambda=\mu$ are easy to verify. The $\psi(s)$ for $s\in[0,1]$ has been computed above.
When $s\le0$, since
$$
{\lambda^i(1-\lambda)^{n-i}+\lambda^{n-i}(1-\lambda)^i\over2}
\le\lambda^i(1-\lambda)^{n-i},\qquad1\le i\le[n/2],
$$
we see from \eqref{F-6.1} that $\Tr\hat\rho_{0,n}^s\hat\rho_{1,n}^{1-s}$ is lower bounded
by
$$
\sum_{i=0}^{[n/2]}(\lambda^i(1-\lambda)^{n-i})^s
\biggl({\mu^i(1-\mu)^{n-i}\over2}\biggr)^{1-s}{n\choose i}.
$$
It is also upper bounded by
\begin{align*}
&\sum_{i=0}^{[n/2]}\biggl({\lambda^i(1-\lambda)^{n-i}\over2}\biggr)^s
(\mu^i(1-\mu)^{n-i})^{1-s}{n\choose i} \\
&\qquad\qquad+\sum_{i=[n/2]+1}^n\biggl({\lambda^{n-i}(1-\lambda)^i\over2}\biggr)^s
(\mu^{n-i}(1-\mu)^i)^{1-s}{n\choose i} \\
&\qquad\le2\sum_{i=0}^{[n/2]}\biggl({\lambda^i(1-\lambda)^{n-i}\over2}\biggr)^s
(\mu^i(1-\mu)^{n-i})^{1-s}{n\choose i}.
\end{align*}
Hence, by Lemma \ref{L-A.5} of the Appendix we have
\begin{align*}
\psi(s)&=\lim_{n\to\infty}{1\over n}\log\sum_{i=0}^{[n/2]}
{n\choose i}(\lambda^s\mu^{1-s})^i((1-\lambda)^s(1-\mu)^{1-s})^{n-i} \\
&=\begin{cases}
\log(\lambda^s\mu^{1-s}+(1-\lambda)^s(1-\mu)^{1-s}) & \text{if $s^*\le s\le0$}, \\
{s\over2}\log\lambda(1-\lambda)+{1-s\over2}\log\mu(1-\mu)+\log2 & \text{if $s\le s^*$},
\end{cases}
\end{align*}
where $s^*\in(-\infty,0]$ satisfies
$$
\lambda^{s^*}\mu^{1-s^*}=(1-\lambda)^{s^*}(1-\mu)^{1-s^*}\quad\mbox{or}\quad
\biggl({(1-\lambda)\mu\over\lambda(1-\mu)}\biggr)^{s^*}={\mu\over1-\mu}.
$$
When $s\ge1$, the computation using Lemma \ref{L-A.5} is similar. Summing up all, we write
$$
\psi(s)=\begin{cases}
\log(\lambda^s\mu^{1-s}+(1-\lambda)^s(1-\mu)^{1-s}) & \text{if $s\ge s^*$}, \\
{s\over2}\log\lambda(1-\lambda)+{1-s\over2}\log\mu(1-\mu)+\log2 & \text{if $s\le s^*$},
\end{cases}
$$
which shows the differentiability of $\psi(s)$ at any $s\in\bR$ including $s=s^*$.
\end{example}

The next example shows that in the settings of the Chernoff and the Hoeffding bounds
the restricted measurements may yield a strictly worse performance even if $\rho_1$ is
$G$-invariant.

\begin{example}\label{E-6.2}\rm{\bf (A pure state vs.\ an invariant mixed state with
$\bT$-symmetry)}\quad
Let $\hil:=\bC^2$ and the states to discriminate be
$\hat\rho_0:=\bmatrix1/2&1/2\\1/2&1/2\endbmatrix$ and
$\hat\rho_1:=\bmatrix\alpha&0\\0&1-\alpha\endbmatrix$ with $0<\alpha<1$. In the
unrestricted scenario, we have 
\begin{equation*}
\psin(s):=\log\Tr\hat\rho_0^s\hat\rho_1^{1-s}
=\log{\alpha^{1-s}+(1-\alpha)^{1-s}\over2}.
\end{equation*}
Since
\begin{equation*}
{d\over ds}(\alpha^{1-s}+(1-\alpha)^{1-s})
=-\alpha^{1-s}\log\alpha-(1-\alpha)^{1-s}\log(1-\alpha)>0,
\end{equation*}
we get 
\begin{equation*}
\chdist{\rho_0}{\rho_1}=-\min_{0\le s\le1}\psin(s)=-\psin(0)=\log 2.
\end{equation*} 
Hence, by \eqref{F-4.4},
\begin{equation}\label{F-6.5}
\lim_{n\to\infty}{1\over n}\log P_{\min}(\rho_0^{\otimes n}:\rho_1^{\otimes n})
=-\chdist{\rho_0}{\rho_1}=-\log2.
\end{equation}

Now let $G:=\bT=\{\zeta\in\bC:|\zeta|=1\}$ and define the unitary representation
$u_\zeta:=\bmatrix1&0\\0&\zeta\endbmatrix,\,\zeta\in\bT$, on $\hil$.
It is easy to see that $u_\zeta^{\otimes n}$ is diagonal with $1,\zeta,\ldots,\zeta^n$
standing in the diagonal entries, and $\zeta^i$ appears exactly ${n\choose i}$ times.
Hence,
\begin{equation*}
\cA_n=\bigoplus_{i=0}^nM_{n\choose i}, 
\end{equation*}
and one can also see that 
the Bratteli diagrams of the inclusions $\bC I\subset\cA_1\subset\cA_2\subset\cdots$
form the Pascal triangle. Note that $\rho_1$ is $G$-invariant. The $G$-invariant
reductions of $\rho_0^{\otimes n}$ and $\rho_1^{\otimes n}$ are given by 
\begin{equation}\label{F-6.6}
\hat\rho_{0,n}=E_{\cA_n}(\hat\rho_0^{\otimes n})
=\sum_{i=0}^n{n\choose i}{1\over2^n}P_{n,i},\qquad
\hat\rho_{1,n}=\hat\rho_1^{\otimes n}
=\sum_{i=0}^n\alpha^i(1-\alpha)^{n-i}E_{n,i},
\end{equation}
where $E_{n,i}$ is the identity of $M_{n\choose i}$ with $\sum_{i=0}^nE_{n,i}=I$, and
$P_{n,i}$ is a rank one projection with $P_{n,i}\le E_{n,i}$. Therefore,
$$
\Tr\hat\rho_{0,n}^s\hat\rho_{1,n}^{1-s}=\sum_{i=0}^n
\biggl({n\choose i}{1\over2^n}\biggr)^s(\alpha^i(1-\alpha)^{n-i})^{1-s},
\qquad s\in\bR.
$$
Take $a=\alpha^{1-s}$ and $b=(1-\alpha)^{1-s}$ in Lemma \ref{L-A.4} of the Appendix to
obtain
$$
\psi(s)=\begin{cases}
s\log\bigl(\alpha^{1-s\over s}+(1-\alpha)^{1-s\over s}\bigr)-s\log2
& \text{if $s>0$}, \\
(1-s)\log\max\{\alpha,1-\alpha\}-s\log2 & \text{if $s\le0$}.
\end{cases}
$$
It is obvious that $\psi$ is differentiable at any $s\ne0$. To check the
differentiability at $s=0$, assume $\alpha>1-\alpha$ and set
$\beta:=(1-\alpha)/\alpha\in(0,1)$. Then, for $s>0$ we have
$$
\psi(s)=(1-s)\log\alpha+s\log\bigl(1+\beta^{1-s\over s}\bigr)-s\log2,
$$
and the differentiability at $s=0$ follows from
$$
\lim_{s\searrow0}{s\log\bigl(1+\beta^{1-s\over s}\bigr)\over s}
=\lim_{s\searrow0}\log\bigl(1+\beta^{1-s\over s}\bigr)=0.
$$
The case $\alpha<1-\alpha$ goes in the same way, and the case $\alpha=1/2$ is easy to
verify. Consequently, $\psi$ exists and is differentiable on the whole real line, and
hence the identities of Remark \ref{rem:differentiability} hold. Moreover,
by Lemma \ref{hoeffding relentr} and Theorem \ref{T-5.1},
\begin{equation*}
\msr{\rho_0}{\rho_1}=S(\rho_0\,\|\,\rho_1)
=\psi'(1)=-{\log\alpha+\log(1-\alpha)\over2}
\end{equation*}
so that the error exponents for Stein's lemma are the same in the unrestricted and the
$G$-invariant cases.

If $\alpha=1/2$ (i.e., $\hat\rho_1=2^{-1}I_2$), then $\psi(s)=-(1-s)\log 2=\psin(s)$ for
all $s\ge0$ so that  $\chmdist{\rho_0}{\rho_1}=\chdist{\rho_0}{\rho_1}$ and
$\hmdist{\rho_0}{\rho_1}{r}=H(r\,|\,\rho_0\,\|\,\rho_1)$ for all $r\ge0$.
Therefore, by \eqref{F-4.10} and \eqref{F-6.5},
$$
\lim_{n\to\infty}{1\over n}\log P_{\min}(\rho_{0,n}:\rho_{1,n})=-\log2
=\lim_{n\to\infty}{1\over n}\log P_{\min}(\rho_0^{\otimes n}:\rho_1^{\otimes n}).
$$
Assume for the rest that $\alpha\ne 1/2$. Then
$$
\psi(0)=\log\max\{\alpha,1-\alpha\}>-\log2=\psin(0),\qquad\psi(1)=0=\psin(1),
$$
and for any $s\in(0,1)$,
\begin{equation*}
\psi(s)=\log\Biggl({\alpha^{1-s\over s}+(1-\alpha)^{1-s\over s}\over2}\Biggr)^s
>\log{\alpha^{1-s}+(1-\alpha)^{1-s}\over2}=\psin(s)
\end{equation*}
thanks to the strict concavity of $x\mapsto x^s$, $x\ge0$. Therefore, again by
\eqref{F-4.10} and \eqref{F-6.5},
\begin{equation*}
\lim_{n\to\infty}{1\over n}\log P_{\min}(\rho_{0,n}:\rho_{1,n})
=-\chmdist{\rho_0}{\rho_1}>-\chdist{\rho_0}{\rho_1}
=\lim_{n\to\infty}{1\over n}\log P_{\min}(\rho_0^{\otimes n}:\rho_1^{\otimes n}).
\end{equation*}
Furthermore, for any $r>-\psi(1)=0$, by \eqref{F-3.10} we have
$\hmdist{\rho_0}{\rho_1}{r}=(-s_0r-\psi(s_0))/(1-s_0)$ for some $s_0\in[0,1)$ (see the
proof of Lemma \ref{L-A.2}), and hence
$$
\hmdist{\rho_0}{\rho_1}{r}<{-s_0r-\psin(s_0)\over1-s_0}
\le H(r\,|\,\rho_0\,\|\,\rho_1).
$$
Remark \ref{rem:differentiability} then implies that
\begin{align*}
&\hli{\rho_0}{\rho_1}{r}=\hls{\rho_0}{\rho_1}{r}
=\hl{\rho_0}{\rho_1}{r}=-\hmdist{\rho_0}{\rho_1}{r} \\
&\quad>\underline h(r\,|\,\rho_0\,\|\,\rho_1)=\overline h(r\,|\,\rho_0\,\|\,\rho_1)
=h(r\,|\,\rho_0\,\|\,\rho_1)=-H(r\,|\,\rho_0\,\|\,\rho_1),\ds\ds r>0.
\end{align*}
That is, in both settings of the Chernoff and the Hoeffding bounds, the optimal
performance of the $G$-invariant measurements is strictly worse than that of the
unrestricted ones. Moreover, since the alternative hypothesis in this example is
$G$-invariant, one sees immediately that the inequalities
\begin{equation*}
\chmdist{\rho_0}{\rho_1}\le
\inf_{g\in G}\chdist{\rho_0}{\rho_1\circ\Ad u_g},\ds\ds
\hmdist{\rho_0}{\rho_1}{r}\le
\inf_{g\in G}\hdist{\rho_0}{\rho_1\circ\Ad u_g}{r}
\end{equation*}
of \eqref{chineq} and \eqref{hineq} cannot hold as an equality in general.
\end{example}

\begin{remark}\label{R-6.3}\rm
Note that if we reduce the group $G$ in the above example, then the difference between the
optimal performances of the restricted and the unrestricted measurements may disappear.
Indeed, consider the subgroup $G:=\bZ_2=\{\pm1\}$ of $\bT$ with the same representation
as in Example \ref{E-6.1}. Let $\rho_0$ and $\rho_1$ be the same as in Example \ref{E-6.2}.
Then $\cA_n=M_{2^{n-1}}\oplus M_{2^{n-1}}$ as in Example \ref{E-6.1} (but with a
different arrangement of basis) and $\rho_{1,n}$ is the same as in \eqref{F-6.6} while
$\rho_{0,n}$ is given by
$$
\hat{\rho}_{0,n}=\biggl({1\over2^n}J_{2^{n-1}}\biggr)
\oplus\biggl({1\over2^n}J_{2^{n-1}}\biggr),
$$
where $J_{2^{n-1}}$ is the $2^{n-1}\times2^{n-1}$ matrix of all entries equal to one.
Since
$$
\hat{\rho}_{0,n}^s=\biggl({1\over2^{s+n-1}}J_{2^{n-1}}\biggr)
\oplus\biggl({1\over2^{s+n-1}}J_{2^{n-1}}\biggr),
$$
we have
$$
\Tr\hat{\rho}_{0,n}^s\hat{\rho}_{1,n}^{1-s}
={1\over2^{s+n-1}}\sum_{i=0}^n{n\choose i}(\alpha^{n-i}(1-\alpha)^i)^{1-s}
={(\alpha^{1-s}+(1-\alpha)^{1-s})^n\over2^{s+n-1}}.
$$
Hence the function $\psi$ in this case is equal to $\psin$:
$$
\psi(s)=\lim_{n\to\infty}{1\over n}\log\Tr\hat\rho_{0,n}^s\hat\rho_{1,n}^{1-s}
=\log{\alpha^{1-s}+(1-\alpha)^{1-s}\over2}=\psin(s),
$$
and therefore, by Remark \ref{rem:differentiability},
\begin{equation*}
\lim_{n\to\infty}{1\over n}\log P_{\min}(\rho_{0,n}:\rho_{1,n})
=\lim_{n\to\infty}{1\over n}\log P_{\min}(\rho_0^{\otimes n}:\sigma_0^{\otimes n}).
\end{equation*}
\end{remark}

\begin{remark}\label{R-6.4}\rm
In the setting of Example \ref{E-6.2} we have seen that
$$
\min_{0\le s\le1}\psi(s)\le\psi(1/2)=-{1\over2}\log2
={1\over2}\min_{0\le s\le1}\psin(s).
$$
On the other hand,
$$
\psi'(s)=\log\bigl(\alpha^{1-s\over s}+(1-\alpha)^{1-s\over s}\bigr)
-{1\over s}\cdot{\alpha^{1-s\over s}\log\alpha+(1-\alpha)^{1-s\over s}\log(1-\alpha)
\over\alpha^{1-s\over s}+(1-\alpha)^{1-s\over s}}-\log2
$$
and hence
$$
\psi'(1/2)=-2(\alpha\log\alpha+(1-\alpha)\log(1-\alpha))-\log2.
$$
This shows that there exists an $\alpha^*\in(0,1)$ such that $\psi'(1/2)=0$ for
$\alpha=\alpha^*$ (a numerical computation shows $\alpha^*\approx0.11$). Hence, if
$\hat\rho_1=\bmatrix\alpha^*&0\\0&1-\alpha^*\endbmatrix$,
then $\psi'(1/2)=0$, and $\psi$ takes the minimum at $s=1/2$ so that
$$
\chmdist{\rho_0}{\rho_1}=\min_{0\le s\le1}\psi(s)
={1\over2}\min_{0\le s\le1}\psin(s)={1\over2}\chdist{\rho_0}{\rho_1}.
$$
Comparing this with \eqref{F-4.10}, we see that the constant $1/2$ in Theorem \ref{T-5.5}
is the best possible. 
\end{remark}

Finally, we consider the discrimination problem of two pure states of a spin-$\frac{1}{2}$
system with $G$-invariant measurements, with $G$ given as in Example \ref{E-6.2}.

\begin{example}\label{E-6.5}\rm{\bf(Two pure states with $\bT$-symmetry)}\quad
In the same setting as in Example \ref{E-6.2}, let $\rho_0$ and $\rho_1$ be pure states
with densities
$$
\hat\rho_0:=\bmatrix\lambda&\sqrt{\lambda(1-\lambda)} \\
\sqrt{\lambda(1-\lambda)}&1-\lambda\endbmatrix,\qquad
\hat\rho_1:=\bmatrix\mu&\sqrt{\mu(1-\mu)}\\\sqrt{\mu(1-\mu)}&1-\mu\endbmatrix,
$$
where $\lambda,\mu\in(0,1)$, $\lambda\ne\mu$. By looking at how $\cA_n$ in Example
\ref{E-6.2} is decomposed into the direct summands $M_{n\choose i}$ for $0\le i\le n$ and
also by looking at the entries of $\hat\rho_0^{\otimes n}$ and $\hat\rho_1^{\otimes n}$,
it is easy to see that
\begin{equation*}
\hat\rho_{0,n}=\bigoplus_{i=0}^n\lambda^{n-i}(1-\lambda)^iJ_{n\choose i},\qquad
\hat\rho_{1,n}=\bigoplus_{i=0}^n\mu^{n-i}(1-\mu)^iJ_{n\choose i},
\end{equation*}
where $J_{n\choose i}$ is the ${n\choose i}\times{n\choose i}$ matrix of all entries
equal to one. Therefore, for any $s\in\bR$,
\begin{align*}
\Tr\hat\rho_{0,n}^s\hat\rho_{1,n}^{1-s}
&=\sum_{i=0}^n(\lambda^{n-i}(1-\lambda)^i)^s(\mu^{n-i}(1-\mu)^i)^{1-s}
\,\Tr J_{n\choose i} \\
&=\sum_{i=0}^n{n\choose i}(\lambda^s\mu^{1-s})^{n-i}((1-\lambda)^s(1-\mu)^{1-s})^i \\
&=\bigl(\lambda^s\mu^{1-s}+(1-\lambda)^s(1-\mu)^{1-s}\bigr)^n.
\end{align*}
Hence we have
$$
\psi(s)={1\over n}\log\Tr\hat\rho_{0,n}^s\hat\rho_{1,n}^{1-s}
=\log\bigl(\lambda^s\mu^{1-s}+(1-\lambda)^s(1-\mu)^{1-s}\bigr)
$$
for all $s\in\bR$ and all $n\in\bN$. Consequently, $\psi$ is differentiable on $\bR$ and
by \eqref{F-4.10},
$$
\lim_{n\to\infty}{1\over n}\log P_{\min}(\rho_{0,n}:\rho_{1,n})
=\log\min_{0\le s\le1}\bigl(\lambda^s\mu^{1-s}+(1-\lambda)^s(1-\mu)^{1-s}\bigr).
$$

On the other hand, since
$$
\psin(s):=\log\Tr\hat\rho_0^s\hat\rho_1^{1-s}
=\log\Tr\hat\rho_0\hat\rho_1
=\log\bigl(\sqrt{\lambda\mu}+\sqrt{(1-\lambda)(1-\mu)}\bigr)^2,
$$
we have
$$
\lim_{n\to\infty}{1\over n}\log P_{\min}(\rho_0^{\otimes n}:\rho_1^{\otimes n})
=2\log\bigl(\sqrt{\lambda\mu}+\sqrt{(1-\lambda)(1-\mu)}\bigr).
$$
We notice that
\begin{equation}\label{F-6.7}
\min_{0\le s\le1}\bigl(\lambda^s\mu^{1-s}+(1-\lambda)^s(1-\mu)^{1-s}\bigr)
>\bigl(\sqrt{\lambda\mu}+\sqrt{(1-\lambda)(1-\mu)}\bigr)^2
\end{equation}
so that
$$
\lim_{n\to\infty}{1\over n}\log P_{\min}(\rho_{0,n}:\rho_{1,n})
>\lim_{n\to\infty}{1\over n}\log P_{\min}(\rho_0^{\otimes n}:\rho_1^{\otimes n})
$$
similarly to Example \ref{E-6.2}. An elementary proof of \eqref{F-6.7} is as follows: Since
\begin{align}
&\bigl(\lambda^s\mu^{1-s}+(1-\lambda)^s(1-\mu)^{1-s}\bigr)
\bigl(\lambda^{1-s}\mu^s+(1-\lambda)^{1-s}(1-\mu)^s\bigr) \nonumber\\
&\qquad\qquad-\bigl(\sqrt{\lambda\mu}+\sqrt{(1-\lambda)(1-\mu)}\bigr)^2 \nonumber\\
&\qquad=\bigl(\lambda^{s\over2}(1-\lambda)^{1-s\over2}\mu^{1-s\over2}(1-\mu)^{s\over2}
-\lambda^{1-s\over2}(1-\lambda)^{s\over2}\mu^{s\over2}(1-\mu)^{1-s\over2}\bigr)^2\ge0,
\label{F-6.8}
\end{align}
and
\begin{equation}\label{F-6.9}
\lambda^{1-s}\mu^s+(1-\lambda)^{1-s}(1-\mu)^s\le1,
\end{equation}
we get $\lambda^s\mu^{1-s}+(1-\lambda)^s(1-\mu)^{1-s}\ge
\sqrt{\lambda\mu}+\sqrt{(1-\lambda)(1-\mu)}$. Furthermore, since $\lambda\ne\mu$,
the equality holds in \eqref{F-6.8} only if $s=1/2$, but inequality \eqref{F-6.9} is
strict when $s=1/2$. (An extension of the inequality to the matrix case is known in
\cite[Theorem 6]{ANSV}.) 

Note also that we have $\psi(0)=\psi(1)=0$ and
$$
\psi'(0)=-S((\mu,1-\mu)\,\|\,(\lambda,1-\lambda)),\qquad
\psi'(1)=S((\lambda,1-\lambda)\,\|\,(\mu,1-\mu)),
$$
where $S((\lambda,1-\lambda)\,\|\,(\mu,1-\mu))$ is the relative entropy of
$(\lambda,1-\lambda)$ and $(\mu,1-\mu)$. We have
\begin{align*}
\sr{\rho_{0,n}}{\rho_{1,n}}
&=\sum_{k=0}^n{n\choose k}\lambda^{n-k}(1-\lambda)^k\log
{{n\choose k}\lambda^{n-k}(1-\lambda)^k\over{n\choose k}\mu^{n-k}(1-\mu)^k} \\
&=\sum_{k=0}^n{n\choose k}\lambda^{n-k}(1-\lambda)^k
\biggl\{(n-k)\log{\lambda\over\mu}+k\log{1-\lambda\over1-\mu}\biggr\} \\
&=n\lambda\log{\lambda\over\mu}+n(1-\lambda)\log{1-\lambda\over1-\mu}
\end{align*}
for all $n\in\bN$ and hence,
$$
\msr{\rho_0}{\rho_1}
=S((\lambda,1-\lambda)\,\|\,(\mu,1-\mu))=\psi'(1).
$$
On the other hand,
$S(\rho_0\,\|\,\rho_1)=+\infty> \msr{\rho_0}{\rho_1}$. This shows again that Theorem
\ref{T-5.1} cannot generally hold, that is, the assumption $\hat\rho_1\in\cA_1$
cannot be removed.
\end{example}

\section{Concluding remarks and problems}

The asymptotic binary state discrimination problem can be formulated in a very
general way, with two sequences $\{\rho_{0,n}\}_{n\in\bN}$ and $\{\rho_{1,n}\}_{n\in\bN}$
to be discriminated. Here, $\rho_{0,n}$ and $\rho_{1,n}$ are states on $\cB(\hil_n)$ of
some Hilbert space $\hil_n$, where $\hil_{n+m}=\hil_n\otimes\hil_m$ need not be assumed;
see, e.g., \cite{NH,HMO}. To get a complete solution of the problems of the Chernoff and
the Hoeffding bounds and of Stein's lemma, i.e., to show the identities of Remark
\ref{rem:differentiability}, the key point is to show that \eqref{F-4.5} holds with
equality for a suitable range of parameters $a$. This can be done, for instance, by
showing that the states satisfy a certain factorization property \cite{HMO1,HMO} or
that the function $\psi$ exists and is differentiable on the interval $(0,1)$
\cite{HMO,MHOF,M}. As the examples of Section \ref{sec:examples} suggest, one can expect
the differentiability of $\psi$ to hold in our setting of group-invariant state
discrimination. The main open question of the present work is to show that this is
indeed the case.

Once the $\psi$-function (and, in an optimal situation, its differentiability) is obtained,
the identities of the error exponents and the corresponding asymptotic statistical
distances follow quite automatically. In this sense, the results of Section
\ref{sec:error bounds} do not depend much on our present setting of group-symmetric
measurements. On the other hand, the determination of the $\psi$-function requires
significantly different techniques in the different scenarios; see, e.g.,
\cite{HMO,MHOF,M} for example. The key technical tool that we used in the present setting
is the monotonicity of quasi-entropies under stochastic maps. It is worthwhile to note that
our analysis works whenever the states to discriminate are defined on a sequence of
algebras $\{\A_n\}_{n\in\N}$ that satisfy $\A_n\otimes\A_m\subset\A_{n+m}$.

Stein's lemma can be considered as an extremal point of the family of the Hoeffding
bounds. In this sense, Stein's lemma is the most asymmetric scenario, which gives a
heuristic support for the validity of Theorem \ref{T-5.1} which says that
$G$-invariance of the measurements does not mean a real restriction as long as the
alternative hypothesis is also $G$-invariant. It is, however, somewhat surprising that
such an asymmetric condition can yield the bound of Theorem \ref{T-5.5} in the totally
symmetric discrimination problem of the Chernoff bound. On the other hand, we conjecture
that if the symmetry group $G$ is finite and the alternative hypothesis is $G$-invariant,
then the function $\psi$ is actually equal to $\psin$ of the unrestricted setting, and
hence $G$-invariant measurements yield the same optimal error exponents as the
unrestricted ones. This would of course also implies that all the inequalities of
\eqref{chineq}--\eqref{steinineq} hold with equality.

\section*{Acknowledgments}

Partial funding was provided by the Grant-in-Aid for Scientific Research (B)17340043
(F.H.); the Grant-in-Aid for JSPS Fellows 18\,$\cdot$\,06916, the Hungarian Research
Grant OTKA T068258 (M.M.); the Grant-in-Aid for Scientific Research on Priority Area
``Deepening and Expansion of Statistical Mechanical Informatics (DEX-SMI)" 18079014,
the MEXT Grant-in-Aid for Young Scientists (A)20686026 (M.H.). This work was also
partially supported by the JSPS Japan-Hungary Joint Project (F.H.\ \& M.H.). Part of this
work was done while M.M.~was a Junior Research Fellow at the Erwin Schr\"odinger Institute
for Mathematical Physics in Vienna.

\appendix
\def\thesection{Appendix \Alph{section}} 
\section{}\label{proofs}
\setcounter{equation}{0}
\def\thesection{\Alph{section}}

The appendix supplies, for the reader's convenience, the details on some technical
lemmas used in the main body of the paper.

\subsection{Lieb's concavity and Ando's convexity}

Let $\cB(\cH)_+$ and $\cB(\cH)_{++}$ denote the set of positive semidefinite and strictly
positive definite operators in $\B(\hil)$, respectively. When $0\le s\le1$, Lieb's
concavity theorem says that the function
\begin{equation}\label{F-A.1}
(A,B)\mapsto\Tr X^*A^sXB^{1-s}
\end{equation}
is jointly concave on $\cB(\cH)_+\times\cB(\cH)_+$ for any $X\in\cB(\cH)$. When
$1\le s\le2$, a complementary result due to Ando \cite[\S4]{An} says that the map
$(A,B)\mapsto A^s\otimes B^{1-s}$ is jointly convex on $\cB(\cH)_{++}\times\cB(\cH)_{++}$.
This joint convexity is equivalently formulated that the function \eqref{F-A.1} is jointly
convex on $\cB(\cH)_{++}\times\cB(\cH)_{++}$ for any $X\in\cB(\cH)$. Indeed, the map
$U:\B(\hil)\to\hil\otimes\hil$, $X\mapsto\sum_iXe_i\otimes e_i$ is a unitary for any
orthonormal basis $\{e_i\}$ of $\hil$ if $\B(\hil)$ is equipped with the Hilbert-Schmidt
inner product, and the two formulations are easily seen to be unitarily equivalent with
any such $U$. Ando's convexity result can slightly be extended to the following:

\begin{lemma}\label{L-A.1}
Let $X\in\cB(\cH)$ be arbitrary and $1\le s\le2$. Let $Q\in\cB(\cH)$ be an orthogonal
projection. Then the function \eqref{F-A.1} is jointly convex on
$\cB(\cH)_+\times\{B\in\cB(\cH)_+:\supp B=Q\}$.
\end{lemma}

\begin{proof}
Since the case $s=1$ is trivial, assume that $1<s\le2$. For $A_1,B_1,A_2,B_2\in\cB(\cH)_+$
with $\supp B_1=\supp B_2=Q$, apply Ando's result to 
$A_{k,\ep}:=A_k+\eps I,\,B_{k,\ep}:=B_k+\eps^{-1}(I-Q)$ to obtain
$$
(\lambda A_{1,\ep}+(1-\lambda)A_{2,\ep})^s
\otimes(\lambda B_{1,\ep}+(1-\lambda)B_{2,\ep})^{1-s}
\le\lambda A_{1,\ep}^s\otimes B_{1,\ep}^{1-s}
+(1-\lambda)A_{2,\ep}^s\otimes B_{2,\ep}^{1-s}
$$
for any $\lambda\in(0,1)$ and $\ep>0$. Taking the limit as $\eps\searrow0$ yields the
assertion since
\begin{align*}
&(\lambda B_{1,\ep}+(1-\lambda)B_{2,\ep})^{1-s} \\
&\quad=(\lambda B_1+(1-\lambda)B_2)^{1-s}+\ep^{s-1}(I-Q)
\longrightarrow(\lambda B_1+(1-\lambda)B_2)^{1-s}
\end{align*}
as well as $B_{k,\ep}^{1-s}\to B_k^{1-s}$.
\end{proof}

\subsection{Mean Hoeffding distance}

\begin{lemma}\label{L-A.2}
Let $\psi$ be a convex function on the interval $[0,1]$ and $\vfi$ be its Legendre-Fenchel
transform, i.e., $\vfi(a):=\sup_{0\le s\le 1}\{as-\psi(s)\}$, $a\in\bR$. Then for every
$r\ge0$,
\begin{equation}\label{F-A.2}
\sup_{a:\,\vfi(a)-a> r}\vfi(a)
=\sup_{0\le s<1}\frac{-sr-\psi(s)}{1-s}.
\end{equation}
\end{lemma}

\begin{proof}
Define
$$
\hat\vfi(a):=\vfi(a)-a=\sup_{0\le s\le 1}\{a(s-1)-\psi(s)\},\qquad a\in\bR.
$$
The following properties can easily be seen by definitions:
\begin{itemize}
\item[(i)] $\vfi$ and $\hat\vfi$ are convex and continuous on $\bR$,
\item[(ii)] $\vfi$ is increasing and $\hat\vfi$ is decreasing on $\R$ with
$\vfi(a)\to+\infty$, $\hat\vfi(a)\to-\psi(1)$ as $a\to+\infty$ and
$\vfi(a)\to-\psi(0)$, $\hat\vfi(a)\to+\infty$ as $a\to-\infty$.
Moreover, $\hat\vfi$ is strictly decreasing on $(-\infty,\derleft{\psi}(1))$. (See
\cite[Lemma 4.1]{HMO}.)
\end{itemize}

First, assume that $r<-\psi(1)$ or that $r=-\psi(1)$ and $\derleft{\psi}(1)=+\infty$.
Then $\hat\vfi(a)>r$ for every $a\in\R$, and the left-hand side of \eqref{F-A.2} is
$+\infty$ since $\lim_{a\to+\infty}\vfi(a)=+\infty$. On the other hand, since
\begin{equation}\label{F-A.3}
\frac{-sr-\psi(s)}{1-s}=\frac{\psi(1)-\psi(s)}{1-s}-\frac{s}{1-s}(r+\psi(1))-\psi(1),
\end{equation}
we have
$$
\lim_{s\nearrow 1}\frac{-sr-\psi(s)}{1-s}=+\infty
$$
and hence the right-hand side of \eqref{F-A.2} is also equal to $+\infty$.

If $r=-\psi(1)$ and $\derleft{\psi}(1)<+\infty$ or if $r>-\psi(1)$, then
there exists an $a_r$ such that $\hat\vfi(a_r)=r$ and $\hat\vfi$ is strictly decreasing
on $(-\infty,a_r]$. Thus, the left-hand side of \eqref{F-A.2} is $\vfi(a_r)$.
Assume that $r=-\psi(1)$ and $\derleft{\psi}(1)<+\infty$. Then $a_r=\derleft{\psi}(1)$
and $\vfi(a_r)=\derleft{\psi}(1)-\psi(1)$. On the other hand, by \eqref{F-A.3}
and the convexity of $\psi$, we have
$$
\sup_{0\le s<1}\frac{-sr-\psi(s)}{1-s}
=\lim_{s\nearrow 1}\frac{-sr-\psi(s)}{1-s}=\derleft{\psi}(1)-\psi(1).
$$
Next, assume that $r>-\psi(1)$, and define
$s_r:=\mathrm{argmax}_{0\le s\le 1}\{a_rs-\psi(s)\}$. Then $s_r<1$ and
$$
r=a_r(s_r-1)-\psi(s_r)\ge a_r(s-1)-\psi(s),\ds\ds\ds 0\le s\le 1,
$$
so that $a_r\ge(-r-\psi(s))/(1-s)$ for any $0\le s<1$ with equality for $s=s_r$. Therefore,
$$
\vfi(a_r)=a_rs_r-\psi(s_r)\ge a_rs-\psi(s)\ge\frac{-sr-\psi(s)}{1-s},\ds\ds\ds 0\le s<1,
$$
and equality holds for $s=s_r$. Hence we see that the right-hand side of \eqref{F-A.2}
is also equal to $\vfi(a_r)$.
\end{proof}

\subsection{Conditional expectation with discrete Weyl operators}

\begin{lemma}\label{L-A.3}
Let $m,d\in\bN$ and $E$ be the conditional expectation (i.e., the partial trace) from
$M_m\otimes M_d$ onto $M_m\otimes I_d$ with respect to the trace. Then there are unitaries
$U_1,\dots,U_{d^2}$ in $I_m\otimes M_d$ such that
$$
E(A)={1\over d^2}\sum_{j=1}^{d^2} U_jAU_j^*
$$
for all $A\in M_m\otimes M_d$.
\end{lemma}

\begin{proof}
Let $\bZ_d$ denote the additive group of $\{0,\ldots,d-1\}$ with the modulo $d$ addition.
Let $\{e_j\,,\,j\in\bZ_d\}$ be a basis in $\bC^d$ and define $Ue_j:=e_{j+1}$ and
$Ve_j:=w^je_j$ for $j\in\bZ_d$, where $w:=e^{i2\pi/d}$. The so defined operators satisfy
the commutation relation $VU=wUV$. Let
$$
W_k:=\bar w^{k_1k_2/2}\,V^{k_1}\,U^{k_2},\qquad k=(k_1,k_2)\in(\bZ_d)^2,
$$
which are the so-called discrete Weyl operators. It can easily be seen that
$$
W_0=I,\quad W_k^*=W_{-k},\quad W_kW_l=w^{(k_1l_2-k_2l_1)/2}\,W_{k+l},
$$
and $\Tr W_k=\delta_{k,0}$ for all $k,l\in\bZ_d^2$. Hence the Weyl operators are unitaries,
and moreover $\mathcal{W}:=\{d^{-1/2}W_k:k\in(\bZ_d)^2\}$ is an orthonormal base
for $M_d$ with respect to the Hilbert-Schmidt inner product. As a consequence, the
commutant of $\mathcal{W}$ is $\bC I$. One can see by a straightforward computation that
for any $A\in M_d$,
$$
\tilde E(A):=\frac{1}{d^2}\sum_{k\in\bZ_d^2}W_kAW_k^*
$$
is in the commutant of $\mathcal{W}$ and hence it is a constant multiple of the identity.
Since $\Tr\tilde E(A)=\Tr A$, we have $\tilde E(A)=d^{-1}(\Tr A)I$. Setting
$U_k:=I_m\otimes W_k$ yields the assertion.
\end{proof}

\subsection{Limiting formulas}

\begin{lemma}\label{L-A.4}
For any $s\in\bR$ and any $a,b\ge0$,
$$
\lim_{n\to\infty}\Biggl(\sum_{i=0}^n{n\choose i}^sa^ib^{n-i}\Biggr)^{1/n}
=\begin{cases}
(a^{1/s}+b^{1/s})^s & \text{if $s>0$}, \\
\max\{a,b\} & \text{if $s\le0$}.
\end{cases}
$$
\end{lemma}

\begin{proof}
If $a=0$ or $b=0$, then the equality is trivial. If $a,b>0$ then we may assume $b=1$ by
homogeneity. Since
$$
\biggl(\max_{0\le i\le n}{n\choose i}^sa^i\biggr)^{1/n}
\le\Biggl(\sum_{i=0}^n{n\choose i}^sa^i\Biggr)^{1/n}
\le(n+1)^{1/n}\biggl(\max_{0\le i\le n}{n\choose i}^sa^i\biggr)^{1/n},
$$
what we have to prove is that
$$
\lim_{n\to\infty}\biggl(\max_{0\le i\le n}{n\choose i}^sa^i\biggr)^{1/n}
=\begin{cases}
(a^{1/s}+1)^s & \text{if $s>0$}, \\
\max\{a,1\} & \text{if $s\le0$}.
\end{cases}
$$
that is,
$$
\lim_{n\to\infty}\max_{0\le i\le n}
\biggl({s\over n}\log{n\choose i}+{i\over n}\log a\biggr)
=\begin{cases}
\log(a^{1/s}+1) & \text{if $s>0$}, \\
\max\{\log a,0\} & \text{if $s\le0$}.
\end{cases}
$$
Since the Stirling formula gives
$$
{s\over n}\log{n\choose i}+{i\over n}\log a
=s\biggl(-{i\over n}\log{i\over n}
-\biggl(1-{i\over n}\biggr)\log\biggl(1-{i\over n}\biggr)+o(1)\biggr)+{i\over n}\log a,
$$
we may show that
$$
\max_{0\le x\le1}h_s(x)
=\begin{cases}
s\log(a^{1/s}+1) & \text{if $s>0$}, \\
\max\{\log a,0\} & \text{if $s\le0$},
\end{cases}
$$
for
$$
h_s(x):=s(-x\log x-(1-x)\log(1-x))+x\log a,\qquad0\le x\le1.
$$
Notice that
$$
h_s'(x)=s\log{1-x\over x}+\log a,\qquad0<x<1.
$$
When $s>0$, the maximizer $x_0$ of $h_s(x)$ satisfies
$$
s\log{1-x_0\over x_0}=-\log a\quad\mbox{or}\quad
{1\over x_0}=1+{1\over a^{1/s}},
$$
and the maximum is
\begin{align*}
h_s(x_0)&=sx_0\log{1-x_0\over x_0}-s\log(1-x_0)+x_0\log a \\
&=-s\log(1-x_0)=s\log{1\over x_0}+\log a=s\log(a^{1/s}+1).
\end{align*}
When $s\le0$, $h_s(x)$ takes the maximum at either $x=0$ or $x=1$, so that the maximum is
$\max\{\log a,0\}$.
\end{proof}

\begin{lemma}\label{L-A.5}
For any $a,b\ge0$,
$$
\lim_{n\to\infty}\Biggl(\sum_{i=0}^{[n/2]}{n\choose i}a^ib^{n-i}\Biggr)^{1/n}
=\begin{cases}
a+b & \text{if $a\le b$}, \\
2\sqrt{ab} & \text{if $a\ge b$}.
\end{cases}
$$
\end{lemma}

\begin{proof}
Since the case $a=0$ or $b=0$ is trivial, we assume that $a,b>0$. As in the proof of Lemma
\ref{L-A.4}, letting $b=1$ we may prove that
$$
\lim_{n\to\infty}\biggl(\max_{0\le i\le[n/2]}{n\choose i}a^i\biggr)^{1/n}
=\begin{cases}
a+1 & \text{if $a\le1$}, \\
2\sqrt a & \text{if $a\ge1$}.
\end{cases}
$$
Hence it suffices to show that
$$
\max_{0\le x\le1/2}h(x)=\begin{cases}
\log(a+1) & \text{if $a\le1$}, \\
\log2\sqrt a & \text{if $a\ge1$},
\end{cases}
$$
where $h(x)$ denotes $h_s(x)$ with $s=1$ in the proof of Lemma \ref{L-A.4}. It is indeed
seen since the maximum of $h(x)$ on $[0,1/2]$ is taken at $x=a/(a+1)$ if $a\le1$ and at
$x=1/2$ if $a\ge1$.
\end{proof}

\end{document}